\pdfoutput=1

\documentclass[sigconf,nonacm]{acmart}
\setcopyright{none}
\renewcommand\footnotetextcopyrightpermission[1]{}
\usepackage[utf8]{inputenc}
\usepackage{graphicx}
\usepackage{xcolor}
\usepackage{amsthm}
\usepackage{amsmath}
\usepackage{preamble}

\newtheorem{theorem}{Theorem}
\newtheorem{corollary}{Corollary}
\newtheorem{lemma}{Lemma}
\newtheorem{proposition}{Proposition}

\newtheorem{definition}{Definition}

\newcommand{\chainnosp}{Babylon\ignorespaces}
\newcommand{\bpow}{Babylon\xspace}
\newcommand{\passive}{passive\xspace}
\newcommand{\chain}{\bpow}

\newcommand{\pos}{PoS\xspace}
\newcommand{\header}{\ensuremath{\mathsf{header}}}
\newcommand{\txroot}{\ensuremath{\mathsf{txr}}}

\newcommand{\LOG}{\ensuremath{\mathsf{Ledger}}}
\newcommand{\Ledger}[2]{%
    \ensuremath{\mathsf{PoSLOG}_{#1}^{#2}}
}

\newcommand{\POWLedger}[2]{%
    \ensuremath{\mathsf{PoWChain}_{#1}^{#2}}
}
\newcommand{\validator}[0]{\ensuremath{\mathsf{v}}}
\newcommand{\client}[0]{\ensuremath{\mathsf{c}}}
\algnewcommand{\algorithmicswitch}{\textbf{switch}}
\algdef{SE}[SWITCH]{Switch}{EndSwitch}[1]{\algorithmicswitch\ #1\ \algorithmicdo}{\algorithmicend\ \algorithmicswitch}%
\algtext*{EndSwitch}%

\algnewcommand{\algorithmiccase}{\textbf{case}}
\algdef{SE}[CASE]{Case}{EndCase}[1]{\algorithmiccase\ #1}{\algorithmicend\ \algorithmiccase}%
\algtext*{EndCase}%

\algnewcommand{\algorithmicon}{\textbf{on}}
\algdef{SE}[ON]{On}{EndOn}[1]{\algorithmicon\ #1\ \algorithmicdo}{\algorithmicend\ \algorithmicon}%
\algtext*{EndOn}%

\algrenewcommand{\algorithmicdo}{}
\algrenewcommand{\algorithmicthen}{}

\algnewcommand{\algorithmicgoto}{\textbf{goto}}%
\algnewcommand{\Goto}[1]{\algorithmicgoto~\ref{#1}}%

\algnewcommand{\algorithmicbreak}{\textbf{break}}%
\algnewcommand{\Break}[0]{\algorithmicbreak}%

\algnewcommand{\algorithmicwaiton}{\textbf{wait on}}%
\algnewcommand{\WaitOn}[1]{\algorithmicwaiton~{#1}}%

\PassOptionsToPackage{hyphens}{url}
\usepackage{url}
\usepackage{hyperref}
\hypersetup{breaklinks=true}

\usepackage[utf8]{inputenc}
\usepackage[T1]{fontenc}

\usepackage{amsfonts,amsmath,amsthm}
\usepackage{graphicx}
\usepackage{xcolor}
\usepackage{lineno}

\usepackage[caption=false,font=footnotesize]{subfig}
\usepackage{float}

\usepackage{IEEEtrantools}
\allowdisplaybreaks

\usepackage[shortlabels,inline]{enumitem}
\setlist[itemize]{leftmargin=5.5mm}
\setlist[enumerate]{leftmargin=5.5mm}

\usepackage{xstring}

\usepackage{datetime2}

\usepackage{siunitx}

\usepackage{varwidth}

\usepackage{tikz}
\usetikzlibrary{calc}
\usetikzlibrary{arrows}
\usetikzlibrary{arrows.meta}
\usetikzlibrary{patterns}
\usetikzlibrary{positioning}
\usetikzlibrary{decorations.pathreplacing}
\usetikzlibrary{shapes.misc}
\usetikzlibrary{spy}

\pgfdeclarelayer{bg1}
\pgfdeclarelayer{bg2}
\pgfsetlayers{bg1,bg2,main}

\usepackage{pgfplots}
\pgfplotsset{compat=1.14}
\usepgfplotslibrary{fillbetween}

\usepackage{xspace}
\usepackage{ifthen}

\newcommand{\alr}{accountable liveness resilience\xspace}

\newcommand{\ssr}{slashable safety resilience\xspace}
\newcommand{\Ssr}{Slashable safety resilience\xspace}

\newcommand{\slr}{slashable liveness resilience\xspace}

\newcommand{\LOGda}[2]{%
    \ifthenelse{\equal{#1}{}}{%
        \ensuremath{\mathsf{LOG}_{\mathrm{da}}^{#2}}%
    }{%
        \ensuremath{\mathsf{LOG}_{\mathrm{da},#1}^{#2}}%
    }%
}
\newcommand{\LOGbft}[2]{%
    \ifthenelse{\equal{#1}{}}{%
        \ensuremath{\mathsf{LOG}_{\mathrm{bft}}^{#2}}%
    }{%
        \ensuremath{\mathsf{LOG}_{\mathrm{bft},#1}^{#2}}%
    }%
}
\newcommand{\LOGacc}[2]{%
    \ifthenelse{\equal{#1}{}}{%
        \ensuremath{\mathsf{LOG}_{\mathrm{acc}}^{#2}}%
    }{%
        \ensuremath{\mathsf{LOG}_{\mathrm{acc},#1}^{#2}}%
    }%
}
\newcommand{\tr}[2]{%
    \ifthenelse{\equal{#1}{}}{%
        \ensuremath{\mathsf{T}^{#2}}%
    }{%
        \ensuremath{\mathsf{T}_{#1}^{#2}}%
    }%
}
\newcommand{\wt}[2]{%
    \ifthenelse{\equal{#1}{}}{%
        \ensuremath{\mathsf{w}^{#2}}%
    }{%
        \ensuremath{\mathsf{w}_{#1}^{#2}}%
    }%
}

\newcommand{\PI}[0]{\ensuremath{\Pi}}

\newcommand{\Adv}[0]{\ensuremath{\mathcal A}}
\newcommand{\Env}[0]{\ensuremath{\mathcal Z}}

\newcommand{\ie}[0]{\emph{i.e.}\xspace}
\newcommand{\eg}[0]{\emph{e.g.}\xspace}
\newcommand{\cf}[0]{\emph{cf.}\xspace}

\newcommand{\tx}[0]{\ensuremath{\mathsf{tx}}}

\newcommand{\Tconfirm}[0]{\ensuremath{T_{\mathrm{fin}}}}

\newcommand{\bprop}[1]{%
    \ifthenelse{\equal{#1}{}}{%
        \ensuremath{\Hat{b}}%
    }{%
        \ensuremath{\Hat{b}_{#1}}%
    }%
}

\newcommand{\ld}[1]{%
    \ifthenelse{\equal{#1}{}}{%
        \ensuremath{\mathrm{L}^{(c)}}%
    }{%
        \ensuremath{\mathrm{L}^{(#1)}}%
    }%
}

\newcommand{\fS}[0]{\ensuremath{f_{\mathrm{s}}}}
\newcommand{\fA}[0]{\ensuremath{f_{\mathrm{a}}}}
\newcommand{\fL}[0]{\ensuremath{f_{\mathrm{l}}}}

\newcommand{\CpReq}[3]{%
    \ifthenelse{\equal{#3}{}}{%
        \ensuremath{\langle\mathsf{#1},#2\rangle}%
    }{%
        \ensuremath{\langle\mathsf{#1},#2\rangle_{#3}}%
    }%
}

\definecolor{myParula01Blue}{RGB}{0,114,189}
\definecolor{myParula02Orange}{RGB}{217,83,25}
\definecolor{myParula03Yellow}{RGB}{237,177,32}
\definecolor{myParula04Purple}{RGB}{126,47,142}
\definecolor{myParula05Green}{RGB}{119,172,48}
\definecolor{myParula06LightBlue}{RGB}{77,190,238}
\definecolor{myParula07Red}{RGB}{162,20,47}

\tikzset{myparula11/.style={color=myParula01Blue,solid,mark=+,mark options={solid}}}
\tikzset{myparula12/.style={color=myParula01Blue,densely dashed,mark=x,mark options={solid}}}
\tikzset{myparula13/.style={color=myParula01Blue,densely dotted,mark=o,mark options={solid}}}
\tikzset{myparula14/.style={color=myParula01Blue,dashdotted,mark=triangle,mark options={solid}}}
\tikzset{myparula15/.style={color=myParula01Blue,dashdotdotted,mark=square,mark options={solid}}}

\tikzset{myparula21/.style={color=myParula02Orange,solid,mark=+,mark options={solid}}}
\tikzset{myparula22/.style={color=myParula02Orange,densely dashed,mark=x,mark options={solid}}}
\tikzset{myparula23/.style={color=myParula02Orange,densely dotted,mark=o,mark options={solid}}}
\tikzset{myparula24/.style={color=myParula02Orange,dashdotted,mark=triangle,mark options={solid}}}
\tikzset{myparula25/.style={color=myParula02Orange,dashdotdotted,mark=square,mark options={solid}}}

\tikzset{myparula31/.style={color=myParula03Yellow,solid,mark=+,mark options={solid}}}
\tikzset{myparula32/.style={color=myParula03Yellow,densely dashed,mark=x,mark options={solid}}}
\tikzset{myparula33/.style={color=myParula03Yellow,densely dotted,mark=o,mark options={solid}}}
\tikzset{myparula34/.style={color=myParula03Yellow,dashdotted,mark=triangle,mark options={solid}}}
\tikzset{myparula35/.style={color=myParula03Yellow,dashdotdotted,mark=square,mark options={solid}}}

\tikzset{myparula41/.style={color=myParula04Purple,solid,mark=+,mark options={solid}}}
\tikzset{myparula42/.style={color=myParula04Purple,densely dashed,mark=x,mark options={solid}}}
\tikzset{myparula43/.style={color=myParula04Purple,densely dotted,mark=o,mark options={solid}}}
\tikzset{myparula44/.style={color=myParula04Purple,dashdotted,mark=triangle,mark options={solid}}}
\tikzset{myparula45/.style={color=myParula04Purple,dashdotdotted,mark=square,mark options={solid}}}

\tikzset{myparula51/.style={color=myParula05Green,solid,mark=+,mark options={solid}}}
\tikzset{myparula52/.style={color=myParula05Green,densely dashed,mark=x,mark options={solid}}}
\tikzset{myparula53/.style={color=myParula05Green,densely dotted,mark=o,mark options={solid}}}
\tikzset{myparula54/.style={color=myParula05Green,dashdotted,mark=triangle,mark options={solid}}}
\tikzset{myparula55/.style={color=myParula05Green,dashdotdotted,mark=square,mark options={solid}}}

\tikzset{myparula61/.style={color=myParula06LightBlue,solid,mark=+,mark options={solid}}}
\tikzset{myparula62/.style={color=myParula06LightBlue,densely dashed,mark=x,mark options={solid}}}
\tikzset{myparula63/.style={color=myParula06LightBlue,densely dotted,mark=o,mark options={solid}}}
\tikzset{myparula64/.style={color=myParula06LightBlue,dashdotted,mark=triangle,mark options={solid}}}
\tikzset{myparula65/.style={color=myParula06LightBlue,dashdotdotted,mark=square,mark options={solid}}}

\tikzset{myparula71/.style={color=myParula07Red,solid,mark=+,mark options={solid}}}
\tikzset{myparula72/.style={color=myParula07Red,densely dashed,mark=x,mark options={solid}}}
\tikzset{myparula73/.style={color=myParula07Red,densely dotted,mark=o,mark options={solid}}}
\tikzset{myparula74/.style={color=myParula07Red,dashdotted,mark=triangle,mark options={solid}}}
\tikzset{myparula75/.style={color=myParula07Red,dashdotdotted,mark=square,mark options={solid}}}

\pgfplotsset{
    mysimpleplot/.style = {
        every axis plot/.prefix style={thick},
        width=1.0\linewidth,
        height=0.75\linewidth,
        title style={font=\scriptsize,align=center},
        legend cell align=left,
        legend style={font=\scriptsize},
        legend columns=3,
        legend style={
            at={(0.5,1)},
            yshift=0.3em,
            anchor=south,
            draw=none,
            /tikz/every even column/.append style={
                column sep=0.3em
            },
            cells={
                align=left
            }
        },
        grid=both,
        minor tick num=3,
        major grid style={solid,draw=gray!50},
        minor grid style={densely dotted,draw=gray!50},
        label style={font=\scriptsize,align=center},
        tick label style={font=\scriptsize},
    },
}

\author{Ertem Nusret Tas}
\affiliation{Stanford University\country{}}
\email{nusret@stanford.edu}

\author{David Tse}
\affiliation{Stanford University\country{}} 
\email{dntse@stanford.edu}

\author{Fisher Yu}
\affiliation{Hash Laboratories\country{}}
\email{fisher.yu@hashlabs.cc}

\author{Sreeram Kannan}
\affiliation{University of Washington, Seattle\country{}}
\email{ksreeram@uw.edu}

\thanks{Contact author: DT}

\begin{document}

\title{\chainnosp: \\Reusing Bitcoin Mining to Enhance Proof-of-Stake Security}

\begin{abstract}
    Bitcoin is the most secure blockchain in the world, supported by the immense hash power of its Proof-of-Work miners, but consumes huge amount of energy. 
    Proof-of-Stake chains are energy-efficient, have fast finality and accountability, but face several fundamental security issues: susceptibility to non-slashable long-range safety attacks, non-slashable transaction censorship and stalling attacks and difficulty to bootstrap new PoS chains from low token valuation. 
    We propose Babylon, a blockchain platform which combines the best of both worlds by reusing the immense Bitcoin hash power to enhance the security of PoS chains. 
    \chain provides a data-available timestamping service, securing PoS chains by allowing them to timestamp data-available block checkpoints, fraud proofs and censored transactions on \chain.
    \chain miners merge mine with Bitcoin and thus the platform has zero additional energy cost. 
    The security of a \chainnosp-enhanced PoS protocol is formalized by a cryptoeconomic security theorem which shows slashable safety and liveness guarantees.
\end{abstract}

\maketitle

\section{Introduction}
\label{sec:introduction}

\subsection{From Proof-of-Work to Proof-of-Stake}

Bitcoin, the most valuable blockchain in the world, is secured by a Proof-of-Work protocol that requires its miners to solve hard math puzzles by computing many random hashes. 
As of this writing, Bitcoin miners around the world are computing in the aggregate roughly $1.4 \times 10^{21}$ hashes per second.
This hash power is the basis of Bitcoin's security, as an attacker trying to rewrite the Bitcoin ledger or censor transactions has to acquire a proportional amount of hash power, making it extremely costly to attack the protocol.  
However, this security also comes at a tremendous energy cost.

Many newer blockchains eschew the Proof-of-Work paradigm in favor of energy-efficient alternatives, the most popular of which is Proof-of-Stake (PoS). 
A prominent example is Ethereum, which is currently migrating from PoW to PoS, a process $6$ years in the making. 
Other prominent PoS blockchains include Cardano, Algorand, Solana, Polkadot, Cosmos Hub and Avalanche among others.
In addition to energy-efficiency, another major advantage of many PoS blockchains is their potential to hold protocol violators accountable and slash their stake as punishment.

\subsection{Proof-of-Stake Security Issues}
\label{sec:pos-insecurities}

Early attempts at proving the security of PoS protocols were made under the assumption that the majority or super-majority of the stake belongs to the honest parties (\eg, \cite{snowwhite,badertscher2018ouroboros, algorand, tendermint}).
However, modern PoS applications such as cryptocurrencies are increasingly run by \emph{economic agents} driven by financial incentives, that are not a priori honest.
To ensure that these agents follow the protocol rules, it is crucial to incentivize honest behavior through economic rewards and punishments for protocol-violating behavior.
Towards this goal, Buterin and Griffith \cite{casper} advocated the concept of \emph{accountable safety}, the ability to identify validators who have provably violated the protocol in the event of a safety violation. 
In lieu of making an unverifiable honest majority assumption, this approach aims to obtain a \emph{cryptoeconomic} notion of security for these protocols by holding protocol violators accountable and slashing their stake, thus enabling an exact quantification of the penalty for protocol violation. 
This {\em trust-minimizing} notion of security is central to the design of important PoS protocols such as Gasper \cite{gasper}, the protocol supporting Ethereum 2.0, and Tendermint \cite{tendermint_thesis}, the protocol supporting the Cosmos ecosystem. However, there are several fundamental limitations to achieving such trust-minimizing cryptoeconomic security for PoS protocols:

\begin{enumerate}
    \item {\bf Safety attacks are not slashable:} 
    While a PoS protocol with accountable safety can identify attackers, slashing of their stake is not always possible. thus implying a lack of \emph{slashable safety}. For example, a long-range history-revision attack can be mounted using old coins after the stake is already withdrawn and therefore cannot be slashed \cite{vitalik_weak_subj, snowwhite, badertscher2018ouroboros, long_range_survey}.
    These attacks are infeasible in a PoW protocol like Bitcoin as the attacker needs to counter the total difficulty of the existing longest chain.
    In contrast, they become affordable in a PoS protocol since the old coins have little value and can be bought by the adversary at a small price. Such long-range attacks is a long-known problem with PoS protocols, and there have been several approaches to deal with them (Section \ref{sec:related-work}). In Section \ref{sec:long-range-attack}, we show a negative result: no PoS protocol can provide slashable safety without {\em external} trust assumptions. A typical external trust assumption used in practice is {\em off-chain social-consensus checkpointing}. But since this type of checkpointing cannot be done very frequently, the stake lock-up period has to be set very long (\eg, $21$ days is a typical lock-up period for Cosmos zones), reducing the liquidity of the system.
    Moreover, social consensus cannot be relied upon in smaller blockchains with an immature community.
    
    \item {\bf Liveness attacks are not accountable or slashable:} 
    Examples of these attacks include protocol stalling and transaction censorship. 
    Unlike safety attacks where adversary double-signs conflicting blocks, such attacks are hard to be held accountable in a PoS protocol. 
    For example, Ethereum 2.0 attempts to hold protocol stalling accountable by slashing non-voting attesters through a process called {\em inactivity leak} \cite{inactivity_leak}. 
    However, as we discuss in Section \ref{sec:accountable-liveness}, an attacker can create an alternative chain and make it public much later, in which the honest attesters would not be voting and would therefore be slashed. 
    Moreover, there is no known mechanism to hold the censoring of specific transactions accountable. 
    In this context, we show in Section \ref{sec:accountable-liveness} that accountable liveness, let alone slashable liveness, is impossible for any PoS protocol without external trust assumptions. 
    
     \item {\bf The bootstrapping problem:} 
     Even if a PoS protocol could provide slashable security guarantees, the maximum financial loss an adversary can suffer due to slashing does not exceed the value of its staked coins.
    Thus, the cryptoeconomic security of a PoS protocol is proportional to its token valuation.
    Many PoS chains, particularly ones that support one specific application, \eg, a Cosmos zone, start small with a low token valuation. This makes it difficult for new blockchains to support high-valued applications like decentralized finance or NFTs.
\end{enumerate}

\subsection{Reusing Bitcoin Mining to Provide \\External Trust}
\label{sec:reusing-hash-power}

\begin{figure}[H]
    \centering
    \includegraphics[width=\linewidth]{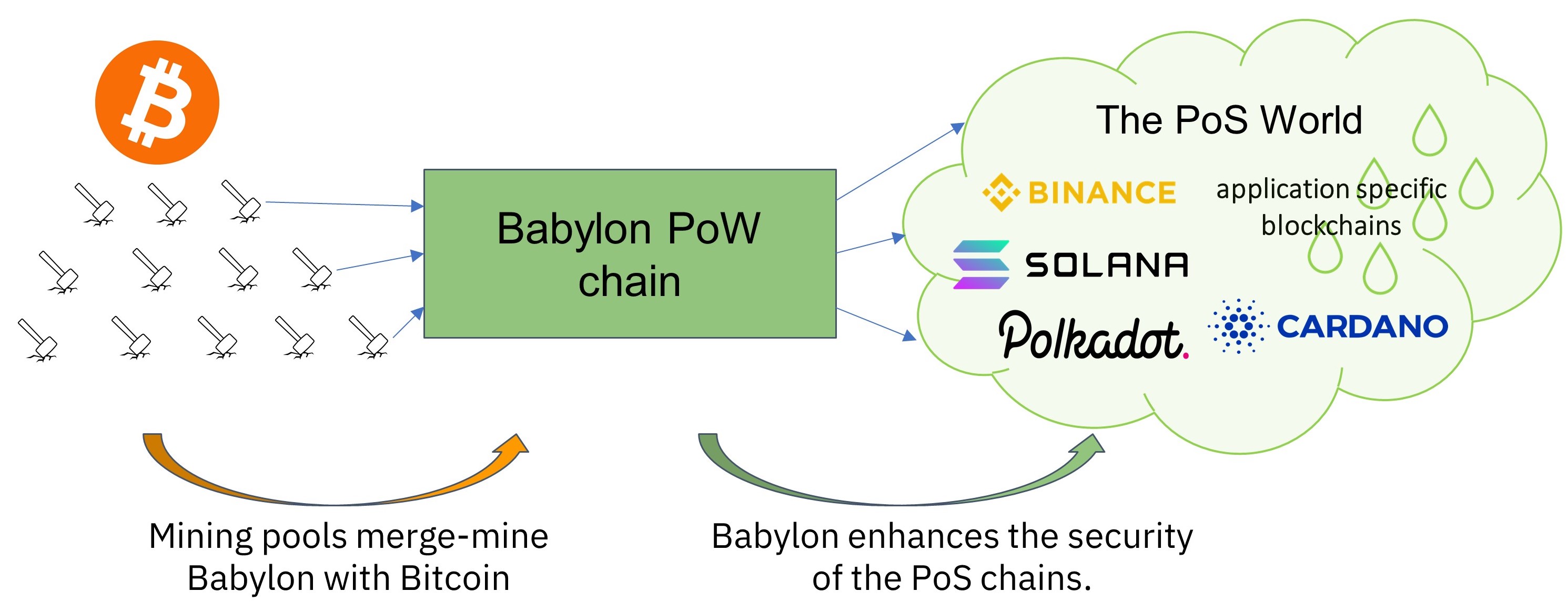}
    \caption{The \chain architecture. \chain is a PoW chain merge mined by Bitcoin miners and used by the PoS protocols to obtain slashable security.}
    \label{fig:arch}
\end{figure}

\begin{figure*}
    \centering
    \includegraphics[width=0.9\linewidth]{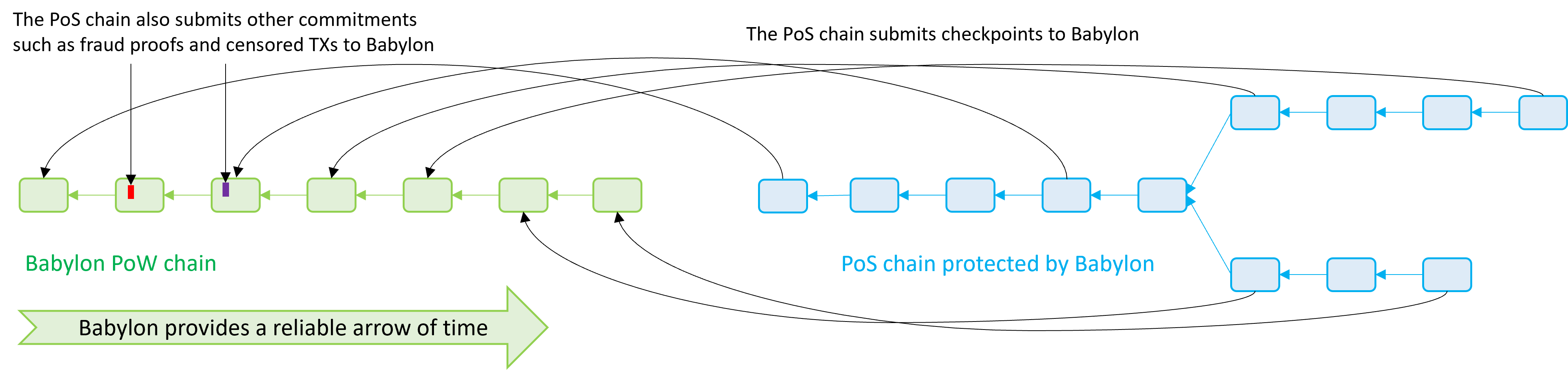}
    \caption{Timestamping on \chainnosp. \chain PoW chain provides a record of the times events happen on the PoS chains, thus enabling PoS nodes to resolve safety violations on the PoS chain. For instance, since the checkpoints of the blocks in the top branch of the PoS chain appears earlier on \chain than the checkpoints of the blocks in the bottom branch, the canonical PoS chain in this case follows the top branch.}
    \label{fig:checkpointing}
\end{figure*}

The PoS security issues above cannot be resolved without an external source of trust. But a strong source of trust already exists in the blockchain ecosystem: Bitcoin mining. Built on this observation, we propose \chainnosp, a blockchain platform which
{\em reuses} the existing Bitcoin hash power to enhance the security for any PoS chain which uses the platform (Figure \ref{fig:arch}). \chain is a PoW
blockchain on which multiple PoS chains can post information
and use the ordering and availability of that information to obtain cryptoeconomic security guarantees while retaining all of their desirable features such as fast finalization.
\chain is mined by the Bitcoin miners via a technique called {\em merge mining} \cite{merge-mining, namecoin, rsk}, which enables the reuse of the same hash power on multiple chains (cf. Appendix \ref{sec:merge-mining}). 
Thus, by reusing the existing Bitcoin hash power, \chain enhances the security of PoS chains at no extra energy cost.

\subsection{Babylon: A Data-available \\ Timestamping Service}
\label{sec:babylon-timestamping-service}

\begin{figure*}
    \centering
    \includegraphics[width=0.8\linewidth]{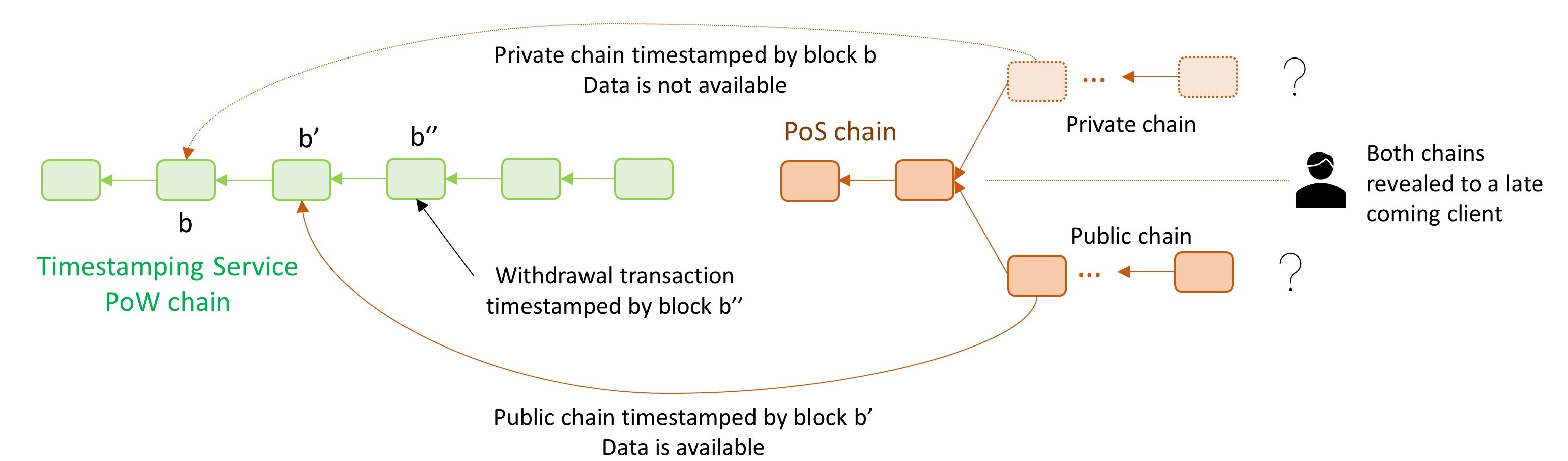}
    \caption{An adversary that controls a super-majority of stake finalizes PoS blocks on two conflicting chains. It keeps the blocks on one of the chains, the attack chain, private, and builds the other, the canonical one in public. It also posts succinct commitments of blocks from both chains to a timestamping service, \eg, Bitcoin or Ethereum. Commitments from the attack chain are ignored by the nodes since their data is not available and might be invalid. Finally, after withdrawing its stake and timestamping it on the PoW chain, the adversary publishes the private attack chain. From the perspective of a late-coming node, the attack chain \emph{is} the canonical one as it has an earlier timestamp on block $b$ and its contents are now available. However, this is a safety violation. Moreover, as the adversary has withdrawn its stake, it cannot be slashed or financially punished.}
    \label{fig:data-availability-attack}
\end{figure*}

A major reason behind the security issues of PoS protocols described in Section~\ref{sec:pos-insecurities} is the lack of a reliable {\em arrow of time}. 
For example, long-range attacks exploit the inability of late-coming nodes to distinguish between the original chain and the adversary's history-revision chain that is publicized much later \cite{snowwhite, long_range_survey}. 
\chain resolves these security limitations by providing a \emph{data-available timestamping service} to the PoS chains.

To obtain slashable security guarantees, full nodes of the PoS protocols post commitments of protocol-related messages, \eg, finalized PoS blocks, fraud proofs, censored PoS transactions, onto the \chain PoW chain (Figure~\ref{fig:checkpointing}). 
\chain checks if the messages behind these commitments are available and provides a timestamp for the messages by virtue of the location of its commitments in the \chain chain.
This enables PoS nodes, including the late-coming ones, to learn the time and order in which each piece of data was first made public.
PoS nodes can then use the timestamps on this data in conjunction with the consensus logic of the native PoS protocol to resolve safety conflicts, identify protocol violators and slash them before they can withdraw their stake in the event of safety or liveness violations.
For example, whenever there is a safety violation in a PoS protocol causing a fork, timestamps on \chain can be used to resolve the fork by choosing the branch with the earlier timestamp (Figure~\ref{fig:checkpointing}). 
Whenever there is a proof of double-signing or liveness violation recorded on \chainnosp, responsible PoS participants can be irrefutably identified and slashed using the information on \chain.

To resolve the security issues of PoS protocols, \chain, {\em in addition} to timestamping PoS data, has to guarantee that this data is {\em available}, \ie, has been publicized to the honest PoS nodes, when it is timestamped.
Otherwise an adversary that controls the majority stake can mount a non-slashable long range attack by posting succinct commitments of finalized, yet private PoS blocks on \chain and releasing the block data to the public much later, after the adversarial stake is withdrawn (Figure~\ref{fig:data-availability-attack}).
Thus, \chain must also provide the additional functionality of checking for the availability of the PoS data it is timestamping.
This functionality cannot be satisfied by solutions that timestamp PoS data by posting its succinct commitments directly on Bitcoin or Ethereum \cite{bms}, whereas posting all of the data raises scalability concerns (\cf Section~\ref{sec:related-work} for more discussion).
Thus, it necessitates a new PoW chain, \chainnosp, whose miners are instructed to check for the availability of the timestamped data in the view of the PoS nodes.

\chain is {\em minimalistic} in the sense that it provides a data-available timestamping service {\em and no more}. 
It does not execute the transactions on the PoS chains, does not keep track of their participants and in fact does not even need to understand the semantics of the PoS block content. 
It also does not store the PoS data. 
All \chain needs to do is to check the availability of the PoS data it is timestamping and make this data public to the PoS nodes, which can be done efficiently (eg. \cite{albassam2018fraud, cmt}). 
This minimalism allows the scalability of \chain to support the security of many PoS chains simultaneously.

\subsection{High-level Description of the Protocol}
\label{sec:protocol-specifics}

The \chain architecture consists of two major components: the \chain PoW chain (\chain for short), merge-mined by Bitcoin miners, hereafter referred to as \chain miners, and the \chainnosp-enhanced \pos protocols each maintained by a distinct set of PoS nodes (nodes for short).
The \chainnosp-enhanced PoS protocol is constructed on top of a standard accountably-safe PoS protocol which takes \pos transactions as input and outputs a chain of \emph{finalized} blocks called the PoS chain.

\bpow miners are \emph{full nodes} towards \chain, \ie, download and verify the validity of all \bpow blocks.
They also download the data whose commitments are sent to \chain for timestamping, but do not validate it.
\pos full nodes download and verify the validity of all \pos blocks in their respective chains.
All \pos nodes act as \emph{light clients} towards \chain.
Thus, they rely on \bpow miners to hand them valid \bpow blocks, and download only the messages pertaining to their \pos protocol from \chain.
A subset of \pos full nodes, called validators, lock their funds for staking and run the \pos consensus protocol by proposing and voting for blocks\footnote{Not every \pos full node is necessarily a validator. Full nodes that are not validators still download \pos blocks and process their transactions to obtain the latest \pos state.}.

\pos nodes can send different types of messages to \chain for timestamping (Figure~\ref{fig:checkpointing}).
These messages are typically succinct commitments of \pos data such as finalized \pos blocks, censored transactions and fraud proofs identifying misbehaving \pos validators.
Finalized \pos blocks whose commitments are included in \chain are said to be checkpointed by the \chain block that includes the commitment.

\pos nodes use the timestamped information on \chain to resolve safety violations and slash protocol violators as described below (\cf Section \ref{sec:protocol} for details):

\begin{enumerate}
    \item {\bf Fork-choice rule:} If there is a fork on the \pos chain due to a safety violation, then the canonical chain of the \chainnosp-enhanced PoS protocol follows the fork whose first checkpointed block has the earlier timestamp on \chain (Figure~\ref{fig:checkpointing}). 
    Thus, \chain helps resolve safety violations on the PoS chains.

    \item {\bf Stake withdrawal and slashing for double-signing:} 
    A validator can input a withdrawal request into the PoS protocol to withdraw its funds locked for staking.
    A stake withdrawal request is granted, as a transaction in a later PoS block, if the PoS block containing the request is timestamped, \ie checkpointed, by a \chain block that is at least $k_w$-block deep in the longest \chain chain and no \emph{fraud proof} of the validator double-signing has appeared on \chain (Figure~\ref{fig:withdrawal}).
    On the other hand, if a fraud proof against this validator exists, then the validator is slashed.
    Here, $k_w$ determines the stake withdrawal delay.
    
    \item {\bf Slashing for transaction censoring:} If a node believes that a transaction is being censored on the PoS chain, it can submit the transaction to \chain along with a \emph{censorship complaint}.
    Upon observing a complaint on the \chain chain, validators stop proposing or voting for PoS blocks that do not contain the censored transaction.
    PoS blocks excluding the censored transaction and checkpointed on \chain after the censorship complaint are labeled as \emph{censoring}.
    Validators that propose or vote for censoring blocks are slashed.
    
    Although \chain is a timestamping service, granularity of time as measured by its blocks depends on the level of its security. 
    For instance, if the adversary can reorganize the last $k_c$ blocks on the \chain chain, it can delay the checkpoints of \pos blocks sent to the miners before a censorship complaint until after the complaint appears on \chain.
    Thus, to avoid slashing honest validators that might have voted for these blocks, \pos blocks checkpointed by the first $k_c$ \chain blocks following a censorship complaint are not labelled as censoring.
    Since this gives the adversary an extra $k_c$ blocktime to censor transactions, $k_c$ is an upper bound on the worst-case finalization latency of transactions when there is an active censorship attack.
    
    \item {\bf Slashing for stalling:} If a \pos node believes that the growth of the PoS chain has stalled due to missing proposals or votes, it submits a \emph{stalling evidence} to \chain. 
    Upon observing a stalling evidence, validators record proposals and votes exchanged over the next round of the \pos protocol on \chain.
    Those that fail to propose or vote, thus whose protocol messages do not appear on \chain within $k_c$ blocks of the stalling evidence, are slashed.
    Again, the grace period of $k_c$ blocks protects honest validators from getting slashed in case the adversary delays their messages on \chain.
\end{enumerate}

We note that the \chain checkpoints of finalized \pos blocks are primarily used to resolve safety violations and slash adversarial validators in the event of safety and liveness attacks.
Hence, \chain does not require any changes to the native finalization rule of the \pos protocols using its services, and preserves their fast finality propert in the absence of censorship or stalling attacks.

Majority of the \chainnosp-specific add-ons used to enhance \pos protocols treat the \pos protocol as a blackbox, thus are applicable to any propose-vote style accountably-safe \pos protocol.
In fact, the only part of the \chainnosp-specific logic in Section~\ref{sec:protocol} that uses the Tendermint details is the part used to slash stalling attacks on liveness.
Hence, we believe that \chain can be generalized to apply to \pos protocols such as PBFT \cite{pbft}, HotStuff \cite{yin2018hotstuff}, and Streamlet \cite{streamlet}.

\subsection{Security Theorem}
\label{sec:sec-analysis}

Using \chain, accountably-safe PoS protocols can overcome the limitations highlighted in Section~\ref{sec:pos-insecurities} and obtain slashable security.
To demonstrate this, we augment Tendermint \cite{tendermint} with \chainnosp-specific add-ons and state the following security theorem for \chainnosp-enhanced Tendermint.
Tendermint was chosen as it provides the standard accountable safety guarantees \cite{tendermint_thesis}.

The \bpow chain is said to be secure for parameter $r$, if the $r$-deep prefixes of the longest \bpow chains in the view of honest nodes satisfy safety and liveness as defined in \cite{backbone}.
A validator $v$ is said to become slashable in the view of an honest node if $v$ was provably identified as a protocol violator and has not withdrawn its stake in the node's view.
Formal definitions of safety and liveness for the \chainnosp-enhanced PoS protocols, slashability for the validators and security for the \bpow chain are given in Section~\ref{sec:model}.

\begin{theorem}
\label{thm:main-security-informal}
Consider a synchronous network where message delays between all nodes are bounded, and the average time between two \chain blocks is set to be much larger than the network delay bound.
Then, \chainnosp-enhanced Tendermint satisfies the following security properties if there is at least one honest \pos node at all times:
\begin{itemize}
\item Whenever the safety of the PoS chain is violated, either of the following conditions must hold:
\begin{itemize}
    \item \textbf{S1:} More than $1/3$ of the active validator set becomes slashable in the view of all honest \pos nodes.
    \item \textbf{S2:} Security of the \bpow chain is violated for parameter $k_w/2$.
\end{itemize}
\item Whenever the liveness of the PoS chain is violated for a duration of more than $\Theta(k_c)$ block-time as measured in mined \chain blocks, either of the following conditions hold:
\begin{itemize}
    \item \textbf{L1:} More than $1/3$ of the active validator set becomes slashable in the view of all honest \pos nodes.
    \item \textbf{L2:} Security of the \bpow chain is violated for parameter $k_c/2$.
\end{itemize}
\end{itemize}
\end{theorem}

Proof of Theorem~\ref{thm:main-security-informal} is given in Appendix~\ref{sec:appendix-main-security}.
Note that this is a {\em cryptoeconomic} security theorem as it explicitly states the slashing cost to the attacker to cause a safety or liveness violation (conditions \textbf{S1} and \textbf{L1} respectively). There is no trust assumption on the PoS validators such as having an honest majority. There are trust assumptions on the Babylon miners (as reflected by conditions \textbf{S2} and \textbf{L2}), but these trust assumptions are also quantifiable in terms of the economic cost of the attacker to acquire the hash power to reorganize certain number of \chain blocks. 

Specific implications of the theorem are:
\begin{enumerate}
    
    \item {\bf Slashable safety:} 
    Conditions \textbf{S1} and \textbf{S2} together say that, when the \pos chain is supported by \chain, the attacker must reorganize $k_w/2$ blocks on \chain if it does not want to be slashed for a safety attack on the PoS chain. 
    Since $k_w$ is the stake withdrawal delay and determines the liquidity of the staked funds, \textbf{S2} quantifies the trade-off between stake liquidity and the attacker's cost. 
    When reorganization cost of \chain is high as is the case for a chain merge-mined with Bitcoin, this trade-off also implies much better liquidity than in the current PoS chains (\eg, 21 days in Cosmos).
    
    \item {\bf Slashable liveness:}
    Conditions \textbf{L1} and \textbf{L2} together say that, with \chainnosp's support, the attacker must reorganize $k_c/2$ blocks on \chain if it does not want to be slashed for a liveness attack on the PoS chain.
    Since $k_c$ is the worst-case latency for the finalization of transactions under an active liveness attack,
    \textbf{L2} quantifies the trade-off between the worst-case latency under attack and the attacker's cost. 
\end{enumerate}

\subsection{Bootstapping New PoS Chains}
\label{sec:bootstrapping}

In a PoS protocol with slashable security, the attack cost is determined by the token value (\cf Section~\ref{sec:pos-insecurities}).
On protocols with low initial valuation, this low barrier to attack pushes away high-valued applications that would have increased the token value.
To break this vicious cycle, new PoS protocols can use \chain as a second layer of finalization.
For instance, \pos nodes can require the checkpoint of a finalized \pos block to become, \eg, $k$ blocks deep in the \chain PoW chain before they consider it finalized.
Then, to violate the security of the $k$-deep and finalized PoS blocks, the attacker must not only forgo its stake due to slashing, but also acquire the hash power necessary to reorganize \chain for $k$ blocks.
For this purpose, it has to control over half of the total hash power for a duration $\Theta(k)$ blocks\footnote{Reorganizing one Bitcoin block costs about USD \$0.5M, as of this writing \cite{attack-cost}. Perhaps more importantly, $0\%$ of this hash power is available on nicehash.com.}\cite{zeta_temporary}.
Thus, by increasing $k$, the attack cost can be increased arbitrarily.
Through this extra protection provided by \chain, newer PoS protocols can attract high value applications to drive up their valuation.

Note that a large $k$ comes at the expense of finalization latency, which no longer benefits from the fast finality of the PoS protocol.
This tradeoff between latency and the parameter $k$ can be made individually or collectively by the PoS nodes in a manner that suits the nodes' or the protocol's security needs.
Moreover, once the valuation of the protocol grows sufficiently large, the parameter $k$ can be decreased in proportion to the slashing costs, and eventually removed altogether, enabling the PoS protocols to regain fast finality after a quick bootstrapping period.

\subsection{Outline}
\label{sec:outline}

Section~\ref{sec:related-work} surveys the related work and analyzes alternative timestamping solutions in terms of their ability to provide slashable security. 
Section~\ref{sec:model} introduces the model and the formal definitions used throughout the paper.
Section~\ref{sec:pos-security} formalizes the impossibility results for slashable safety and accountable liveness of PoS protocols.
Sections~\ref{sec:protocol} and~\ref{sec:scalability-of-the-protocol} give a detailed description of a \chainnosp-enhanced \pos protocol, and discuss \chainnosp's potential for scalability.
Finally, Section~\ref{sec:cosmos} provides a reference design for \chainnosp-enhanced Tendermint using Cosmos SDK.

\section{Related Works}
\label{sec:related-work}

\subsection{Long-range Attacks} 

Among all the PoS security issues discussed in Section \ref{sec:pos-insecurities}, long range history revision attacks is the most well-known, \cite{vitalik_weak_subj, snowwhite, badertscher2018ouroboros, long_range_survey} and several solutions have been proposed:
1) checkpointing via social consensus (\eg, \cite{vitalik_weak_subj, snowwhite,winkle,barber2012}); 
2) use of key-evolving signatures (\eg, \cite{algorand,kiayias2017ouroboros,badertscher2018ouroboros}); 
3) use of verifiable delay functions, \ie, VDFs (\eg, \cite{solana}); 
4) timestamping on an existing PoW chain like Ethereum \cite{bms} or Bitcoin \cite{sarah_talk}.

\subsubsection{Social Consensus}
Social consensus refers to a trusted committee of observers, potentially distinct from the \pos nodes, which periodically checkpoint finalized PoS blocks that have been made available.
It thus attempts to prevent long range attacks by making the adversarial blocks that are kept private distinguishable from those on the canonical PoS chain that contain checkpoints. 

Social consensus suffers from vagueness regarding the size and participants of the checkpointing committee.
For instance, a small oligarchy of trusted nodes would lead to a centralization of trust, anathema to the spirit of distributed systems.
Conversely, a large committee would face the problem of reaching consensus on checkpoints in a timely manner.
Moreover, the question of who belongs in the committee complicates the efforts to quantify the trust assumptions placed on social consensus, in turn making any security valuation prone to miscalculations.
For instance, a re-formulation of Theorem~\ref{thm:main-security-informal} in this setting would claim slashable security as long as the social consensus checkpoints are `trustworthy', without much insight on how to value this trust in economic terms.
In comparison, the trust placed on \chain is quantifiable and equals the cost of acquiring the hash power necessary to reorganize the \bpow PoW chain, which is well-known \cite{attack-cost}.

\subsubsection{Key-evolving Signatures}
Use of key-evolving signatures requires validators to forget old keys so that a history revision attack using old coins cannot be mounted.
However, an adversarial majority can always record their old keys and use them to attack the canonical chain by creating a conflicting history revision chain once they withdraw their stake.
This way, they can cause a safety violation, yet upon detection, avoid any slashing of the stake as it was already withdrawn.
Hence, key-evolving signatures cannot prevent long range attacks without an honest majority assumption, thus cannot provide slashable security.

Security has been shown for various PoS protocols \cite{algorand,badertscher2018ouroboros} using key-evolving signatures under the honest majority assumption, which ensures that the majority of validators willingly forget their old keys.
However, this is not necessarily incentive-compatible as there might be a strong incentive for the validators to remember the old keys in case they become useful later on.
Thus, key-evolving signatures render the honest majority assumption itself questionable by asking honest validators for a favor which they may be tempted to ignore.

\subsubsection{VDFs}
As was the case with key-evolving signatures, VDFs cannot prevent long range attacks without the honest majority assumption, thus cannot provide slashable security. 
For instance, an adversarial majority can build multiple conflicting PoS chains since the beginning of time, and run multiple VDF instances simultaneously for both the public PoS chain and the attack chains that are kept private.
After withdrawing their stakes, these validators can publish the conflicting attack chains with the correct VDF proofs.
Thus, VDFs cannot prevent an adversarial majority from causing a safety violation at no slashing cost.

Another problem with VDFs is the possibility of finding faster functions \cite{fastervdf}, which can then be used to mount a long range attack, even under an honest majority assumption.

\subsubsection{Timestamping on Bitcoin or Ethereum}
Timestamping directly on an existing PoW chain, \eg, Bitcoin \cite{sarah_talk} or Ethereum \cite{bms} suffers from the fact that these PoW chains do not check for the availability of the committed data in the view of the PoS nodes, thus, as a solution, is vulnerable to the attack on Figure~\ref{fig:data-availability-attack}.
To mitigate the attack, either all of the committed PoS data, \eg, all the \pos blocks, must be posted on the PoW chain to guarantee their availability, or an honest majority must certify the timestamps to prevent unavailable PoS blocks from acquiring timestamps on the PoW chain. 

The first mitigation creates scalability issues since in this case, miners must not only verify the availability of the PoS data, which can be done through lightweight methods such as data availability sampling \cite{albassam2018fraud,cmt}, but also store the data of potentially many different PoS protocols indefinitely.

The second mitigation was implemented by \cite{bms} through an Ethereum smart contract which requires signatures from over $2/3$ of the PoS validators to timestamp changes in the validator sets.
Since \cite{bms} assumes honest majority, signatures from $2/3$ of the validators imply that the signed changes in the validator set are due to transactions within available and valid \pos blocks.
However, the second mitigation cannot be used to provide cryptoeconomic security without trust assumptions on the validators.
In contrast, \chain miners are modified to do data availability checks, which enables \chain to rely on the miners themselves rather than an honest majority of PoS validators for data availability. 

\subsection{Hybrid PoW-PoS Protocols}

A \chainnosp-enhanced PoS protocol is an example of a {\em hybrid PoW-PoS protocol}, where consensus is maintained by both PoS validators and PoW miners. One of the first such protocols is the Casper FFG finality gadget used in conjunction with a longest chain PoW protocol \cite{casper}. 
The finality gadget is run by PoS validators as an overlay to checkpoint and finalize blocks in an underlay PoW chain, where blocks are proposed by the miners.
The finality gadget architecture is also used in many other PoS blockchains, such as Ethereum 2.0 \cite{gasper} and Polkadot \cite{stewart2020grandpa}.
\chain can be viewed as a "reverse" finality gadget, where the miners run an overlay PoW chain to checkpoint the underlay PoS chains run by their validators.
Our design of \chain that combines an existing PoW protocol with PoS protocols also leverages off insights from a recent line of work on secure compositions of protocols \cite{ebbandflow,sankagiri_clc, acc_gadget}.

\subsection{Blockchain Scaling Architectures}

Scaling blockchains is a longstanding problem. 
A currently popular solution on Ethereum and other platforms is the shift of transaction execution from a base blockchain to {\em rollups}, which execute state transitions and post state commitments on the blockchain. 
Emerging projects like Celestia \cite{albassam2019lazyledger} take this paradigm further by removing execution entirely from the base blockchain and having it provide only data availability and ordering. 
In contrast, the main goal of the present work is not on scalability but on enhancing existing or new PoS protocols with slashable security. 
While rollups derive their security entirely from the base blockchain, the PoS protocols are autonomous and have their own validators to support their security.
In this context, the main technical challenge of this work is how to design the architecture such that the \chain PoW chain augments the existing security of the PoS protocols with slashable security guarantees.
Nevertheless, to scale up our platform to support many PoS protocols, we can leverage off scaling techniques such as efficient data availability checks \cite{albassam2018fraud,cmt} and sharding \cite{ETHsharding}. More discussions can be found in Section \ref{sec:scalability-of-the-protocol}.

\section{Model}
\label{sec:model}
\emph{Validators:}
\pos nodes that run the \pos consensus protocol are called validators.
Each validator is equipped with a unique cryptographic identity.
Validators are assumed to have synchronized clocks.
%A random oracle serves as a common source of randomness.

There are two sets of validators: \passive and active.
Validators \emph{stake} a certain amount of coins to become active and participate in the consensus protocol.
Although staked coins cannot be spent, active validators can send \emph{withdrawal requests} to withdraw their coins.
Once a withdrawal request by an active validator is finalized by the \pos protocol, \ie, included in the PoS chain, the validator becomes \passive and ineligible to participate in the consensus protocol.
The \passive validator is granted permission to withdraw its stake and spend its funds once a \emph{withdrawal delay} period has passed following the finalization of the withdrawal request.

Let $n$ denote the total number of validators that are active at any given time.
The number of \passive validators is initially zero and grows over time as active validators withdraw their stakes and become \passive.

\emph{Environment and Adversary:}
Transactions are input to the validators by the environment $\Env$.
Adversary $\Adv$ is a probabilistic poly-time algorithm.
$\Adv$ gets to corrupt a certain fraction of the validators when they become active, which are then called \emph{adversarial} validators.
It can corrupt any passive validator.
%, and up to a $\beta \in [0,1]$ fraction of the miners.

Adversarial validators surrender their internal state to the adversary and can deviate from the protocol arbitrarily (Byzantine faults) under the adversary's control.
The remaining validators 
% and $(1-\beta)$ fraction of the miners 
are called \emph{honest} and follow the PoS protocol as specified.

\emph{Networking:} Validators can send each other messages. 
Network is synchronous, \ie $\Adv$ is required to deliver all messages sent between honest validators, miners and nodes within a known upper bound $\Delta$.
% However, $\Delta$ might not be known or used by the \pos protocols as a parameter:
% \begin{definition}
% \label{def:responsiveness}
% (From \cite{hybrid}) A consensus protocol is said to be responsive if its transaction confirmation time depends only on
% the network’s actual delay $\delta$, but not on any apriori known upper-bound $\Delta$ on the network delay.
% \end{definition}

% \emph{Sleeping:} We use the concept of \emph{sleepiness} from \cite{sleepy} to model the dynamic participation of the nodes and miners.
% $\Adv$ chooses, for any given time period, whether a node or a miner is \emph{awake} (\ie, online) or \emph{asleep} (\ie, offline).
% An awake honest node executes the \pos protocol faithfully.
% An asleep honest node does not execute the protocol, and messages that would have arrived during the time period are queued and delivered once the node is awake again.
% Adversarial nodes are always awake.

\emph{Accountability:} We assume that the \pos protocol supported by \chain has an accountable safety resilience of $\fA$ (parameter $d$ as defined in \cite{forensics}), \ie, $\fA$ adversarial validators (that are potentially \passive) are irrefutably identified by all \pos nodes as having violated the protocol in the event of a safety violation, and no honest validator can be identified as a protocol violator.
Moreover, for the culpable validators, \pos nodes can create an irrefutable fraud proof showing that they violated the protocol.

\emph{Safety and Liveness for the \pos Protocols:}
Let $\Ledger{i}{t}$ denote the chain of finalized \pos blocks, \ie, the PoS chain, in the view of a node $i$ at time $t$.
Then, safety and liveness for the \pos protocols are defined as follows:

\begin{definition}
\label{def:pos-security}
Let $\Tconfirm$ be a polynomial function of the security parameter $\sigma$ of the \pos protocol $\PI$.
We say that $\PI$ is $\Tconfirm$-secure if the \pos chain satisfies the following properties:
\begin{itemize}
    \item \textbf{Safety:} For any time slots $t,t'$ and honest \pos nodes $i,j$, either $\Ledger{i}{t}$ is a prefix of $\Ledger{j}{t'}$ or vice versa. For any honest \pos node $i$, $\Ledger{i}{t}$ is a prefix of $\Ledger{i}{t'}$ for all times $t$ and $t'$ such that $t \leq t'$.
    \item \textbf{$\mathbf{\Tconfirm}$-Liveness:} If $\mathcal{Z}$ inputs a transaction $\tx$ to the validators at some time $t$, then, $\tx$ appears at the same position in $\Ledger{i}{t'}$ for any time $t' \geq t+\Tconfirm$ and for any honest \pos node $i$.
\end{itemize}
\end{definition}

A \pos protocol is said to satisfy $\fS$-safety or $\fL$-$\Tconfirm$-liveness if it satisfies safety or $\Tconfirm$-liveness whenever the number of active adversarial validators is less than or equal to $\fS$ or $\fL$ respectively. 

\emph{Safety and Liveness for the \bpow Chain:}
Let $\POWLedger{i}{t}$ denote the longest, \ie, canonical, \bpow chain in the view of a miner or \pos node $i$ at time $t$.
Then, safety and liveness for the \bpow chain are defined as follows:

\begin{definition}[From \cite{backbone}]
\label{def:pow-security}
\chain is said to be secure for parameter $r \geq 1$, $r \in \mathbb{Z}$, if it satisfies the following two properties:
\begin{itemize}
    \item \textbf{Safety:} If a transaction $\tx$ appears in a block which is at least $r$-deep in the longest \bpow chain of an honest node or miner, then, $\tx$ will eventually appear and stay at the same position in the longest \bpow chain of all honest nodes or miners forever.
    \item \textbf{Liveness:} If a valid transaction $\tx$ is received by all honest miners for more than $r$ block-time, then $\tx$ will eventually appear at an $r$-deep block in the longest \bpow chain of all honest nodes or miners. 
    %Discuss this with David
    No invalid transaction ever appears at an $r$-deep block in the longest \bpow chain held by any honest node or miner.
    \end{itemize}
\end{definition}
Thus, if \chain satisfies security for the parameter $r$, $r$-deep prefixes of the longest chains held by the honest nodes are consistent with each other, grow monotonically, and transactions received by the honest miners for more than $r$ block-time enter and stay in the longest \bpow chain observed by the honest nodes forever.

% In the context of \pos protocols augmented with \chainnosp-specific add-ons, $\Ledger{i}{t}$ denotes the finalized \pos ledger or chain (if the protocol is a blockchain protocol) in the view of a \pos node $i$ at time $t$.
% Similarly, $\POWLedger{i}{t}$ denotes the longest \bpow chain in the view of a \pos node $i$ at time $t$.
% $\Ledgerf{i}{t}$ denotes the \pos ledger that have been \emph{finalized} by the \chainnosp-specific add-ons and satisfies $\Ledgerf{i}{t} \preceq \Ledger{i}{t}$ for all \pos nodes $i$ and times $t$ (\cf Section~\ref{sec:protocol}).

\emph{Slashability:} Slashing refers to the process of financial punishment for the active validators detected as protocol violators. 
\begin{definition}
\label{def:slashable}
A validator $\validator$ is said to be \emph{slashable} in the view of a \pos node $\client$ if, 
\begin{enumerate}
    \item $\client$ provably identified $\validator$ as having violated the protocol for the first time at some time $t$, and,
    \item $\validator$ has not withdrawn its stake in $\client$'s view by time $t$.
\end{enumerate} 
\end{definition}
If a validator $\validator$ is observed to be slashable by all honest \pos nodes, no transaction that spends the coins staked by $\validator$ will be viewed as valid by the honest \pos nodes.
%This is because; (i) staked funds could not have been withdrawn in $\client$'s view by the time it identified $\validator$ as a protocol violator, and (ii) $\client$ regards any future attempt by protocol violators to spend their staked funds as invalid.

\section{Impossibility Results for Proof-of-Stake Protocols}
\label{sec:pos-security}

\subsection{Safety Violation is not Slashable}
\label{sec:long-range-attack}

Without additional trust assumptions, \pos protocols are susceptible to various flavors of long range attacks, also known as founders' attack, posterior corruption or costless simulation.
In this context, \cite[Theorem 2]{snowwhite} formally shows that even under the honest majority assumption for the active validators, \pos protocols cannot have safety due to long range attacks without additional trust assumptions.
Since slashable safety is intuitively a stronger result than guaranteeing safety under the honest majority assumption, \cite[Theorem 2]{snowwhite} thus rules out any possibility of providing \pos protocols with slashable safety without additional trust assumptions.
This observation is formally stated by Theorem~\ref{thm:pos-non-slashable} in Appendix~\ref{sec:appendix-proofs}.

\subsection{Liveness Violation is not Accountable}
\label{sec:accountable-liveness}

\begin{figure}
    \centering
    \includegraphics[width=\linewidth]{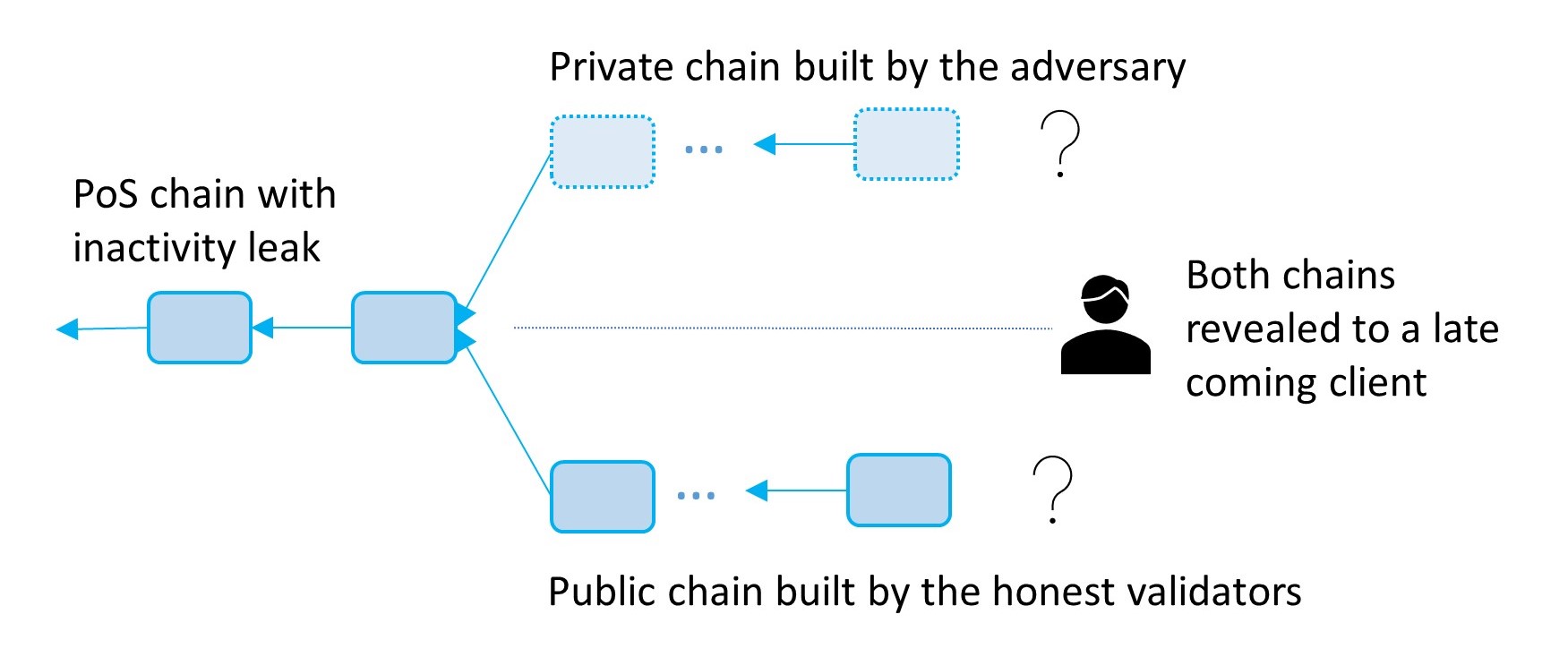}
    \caption{Inactivity leak attack. At the top is adversary's private attack chain. At the bottom is the public canonical chain built by the honest validators.
    Due to inactivity leak, honest \& adversarial validators lose their stake on the attack and canonical chains respectively. A late-coming node cannot differentiate the canonical and attack chains.}
    \vspace{-0.2in}
    \label{fig:inactivity-leak}
\end{figure}

Without additional trust assumptions, \pos nodes cannot identify any validator to have irrefutably violated the \pos protocol in the event of a liveness violation, even under a synchronous environment.
To illustrate the intuition behind this claim, we show that inactivity leak \cite{inactivity_leak}, proposed as a financial punishment for inactive Ethereum 2.0 validators, can lead to the slashing of honest validators' stake with non-negligible probability.
Consider the setup on Figure~\ref{fig:inactivity-leak}, where adversarial validators build a private attack chain that forks off the canonical one and stop communicating with the honest validators.
As honest validators are not privy to the adversary's actions, they cannot vote for the blocks on the attack chain.
Thus, honest validators are inactive from the perspective of the adversary and lose their stake on the attack chain due to inactivity leak.
On the other hand, as the adversarial validators do not vote for the blocks proposed by the honest ones, they too lose their stake on the public, canonical chain (Figure~\ref{fig:inactivity-leak}).
Finally, adversary reveals its attack chain to a late-coming node which observes two conflicting chains.
Although the nodes that have been active since the beginning of the attack can attribute the attack chain to adversarial action, a late-coming node could not have observed the attack in progress. 
Thus, upon seeing the two chains, it cannot determine which of them is the canonical one nor can it irrefutably identify any validator slashed on either chains as adversarial or honest.

To formalize the impossibility of accountable liveness for \pos protocols, we extend the notion of accountability to liveness violations and show that no \pos protocol can have a positive \emph{\alr}, even under a synchronous network with a static set of $n$ active validators that never withdraw their stake.
For this purpose, we adopt the formalism of \cite{forensics} summarized below:
During the runtime of the \pos protocol, validators exchange messages, \eg, blocks or votes, and each validator records its view of the protocol by time $t$ in an execution transcript.
If a node observes that $\Tconfirm$-liveness is violated, \ie, a transaction input to the validators at some time $t$ by $\Env$ is not finalized in the \pos chain in its view by time $t+\Tconfirm$, it invokes a forensic protocol:
The forensic protocol takes transcripts
of the validators as input, and outputs an irrefutable proof that a subset of them have violated the protocol rules.
This proof is sufficient evidence to convince any node, including late-coming ones, that the validators identified by the forensic protocol are adversarial.

Forensic protocol interacts with the nodes in the following way:
Upon observing a liveness violation on the \pos chain, a node asks the validators to send their transcripts.  
It then invokes the forensic protocol with the transcripts received from the validator.
Finally, through the forensic protocol, it constructs the irrefutable proof of protocol violation by the adversarial validators, and broadcasts this proof to all other nodes.

Using the formalization above, we next define \emph{\alr} and state the impossibility theorem for accountable liveness on \pos protocols in the absence of additional trust assumptions:
\begin{definition}
\label{def:accountable-liveness}
$\Tconfirm$-\alr of a protocol is the minimum number $f$ of validators identified by the forensic protocol to be protocol violators when $\Tconfirm$-liveness of the protocol is violated.
Such a protocol provides \emph{$f$-$\Tconfirm$-accountable-liveness}.
\end{definition}

\begin{theorem}
\label{thm:accountable-liveness} 
Without additional trust assumptions, no \pos protocol provides both $\fA$-$\Tconfirm$-accountable-liveness and $\fS$-safety for any $\fA,\fS > 0$ and $\Tconfirm<\infty$.
\end{theorem}

Proof is presented in Appendix~\ref{sec:appendix-proofs} and generalizes the indistinguishability argument for the conflicting chains from the inactivity leak attack.
It rules out any possibility of providing accountable liveness for \pos protocols even under a $\Delta$-synchronous network and a static set of active validators.

A corollary of Theorem~\ref{thm:accountable-liveness} is that \pos protocols cannot have a positive \emph{\slr}:
\begin{definition}
$\Tconfirm$-\slr of a protocol is the minimum number $f$ of validators that are slashable in the view of all \pos nodes per Definition~\ref{def:slashable} when $\Tconfirm$-liveness of the protocol is violated.
Such a protocol provides $f$-$\Tconfirm$-slashable-liveness.
\end{definition}
\begin{corollary}
\label{cor:pos-non-slashable-liveness}
Without additional trust assumptions, no \pos protocol provides both $\fS$-safety and $\fL$-$\Tconfirm$-slahable-liveness for any $\fS,\fL>0$ and $\Tconfirm<\infty$.
\end{corollary}
Proof of Corollary~\ref{cor:pos-non-slashable-liveness} follows from the fact that \alr of \pos protocols is zero without additional trust assumptions.

\section{Protocol}
\label{sec:protocol}
\begin{figure*}[h]
    \centering
    \includegraphics[width=0.9\linewidth]{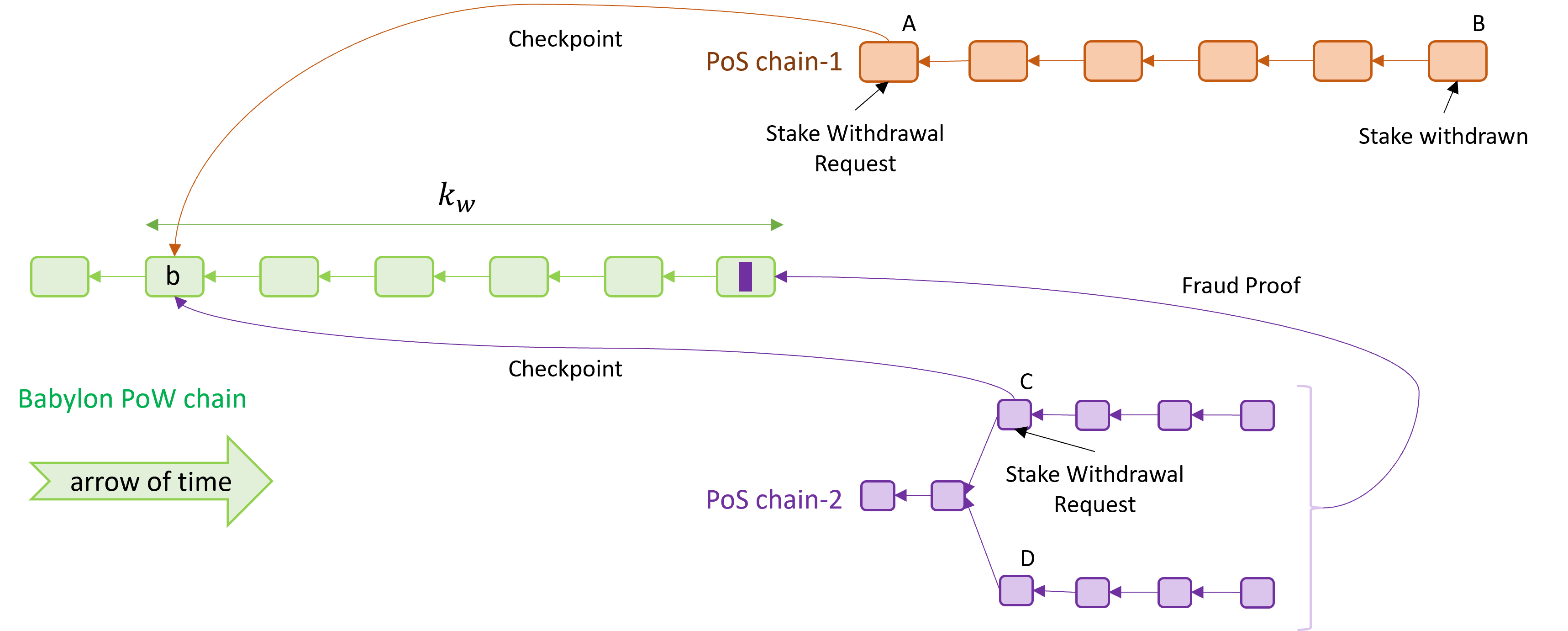}
    \caption{Delayed granting of withdrawal request and slashing. A validator for the PoS chain-1 sends a stake withdrawal request to its chain which is captured by the PoS block A.
    Block A is in turn checkpointed by the \chain block b.
    This stake withdrawal request will only be granted and executed at a later PoS block B, where B is generated by a validator that observes the checkpoint of block A in block b become at least $k_w$ deep in \chain and that there is no fraud proof.
    On the other hand, the validator for the PoS chain-2 is not granted its withdrawal request and is slashed, since a fraud proof appears on \chain before block b becomes $k_w$-deep.
    }
    \label{fig:withdrawal}
\end{figure*}

\begin{figure*}[h]
    \centering
    \includegraphics[width=0.8\linewidth]{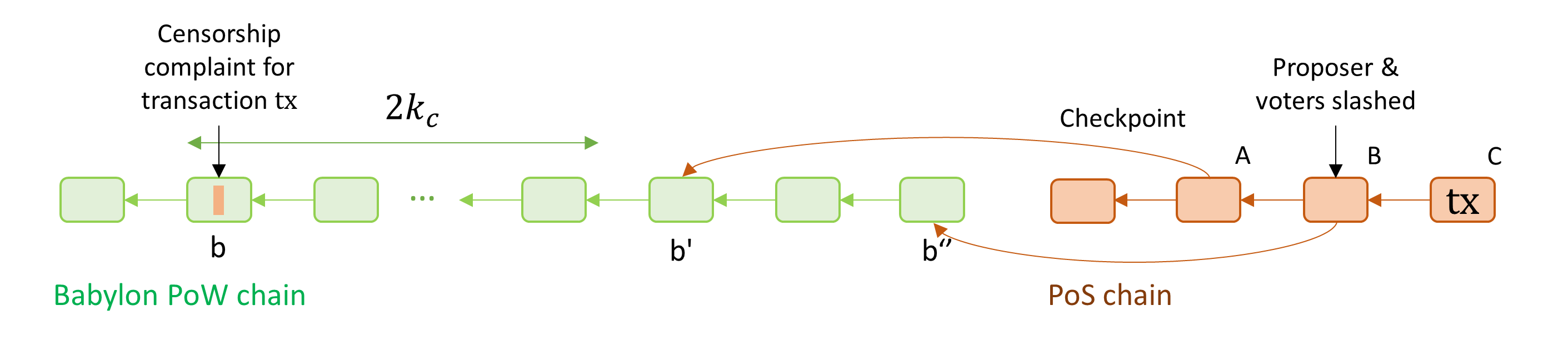}
    \caption{Slashing for censoring. A censored transaction $\tx$ is submitted to \chain through a censorship complaint, and included in a \bpow block $b$. 
    Suppose A is the last \pos block honest nodes proposed or voted for before they observed the censorship complaint on \chain.
    Let $b'$ denote the first \bpow block containing a checkpoint and extending $b$ by at least $2k_c$ blocks.
    Since finalized \pos blocks are checkpointed frequently on \chain, A will be checkpointed by $b'$, or a \bpow block in its prefix.
    Then, any \pos block, \eg, B, checkpointed by a \bpow block following $b'$ must have been proposed or voted upon by the validators after they have observed the censorship complaint, and must include $\tx$.
    However, B, which is checkpointed by $b''$ extending $b'$, does not contain $\tx$, thus, is a censoring block.
    Hence, validators that have proposed and voted for B will be slashed for censorship.
    Here, the $2k_c$ grace period on \chain between $b$ and $b'$ ensures that the honest validators are not slashed for voting upon \pos blocks excluding the censored transactions, before they observed the censorship complaint.}
    \label{fig:censor}
\end{figure*}

\begin{figure*}[h]
    \centering
    \includegraphics[width=0.8\linewidth]{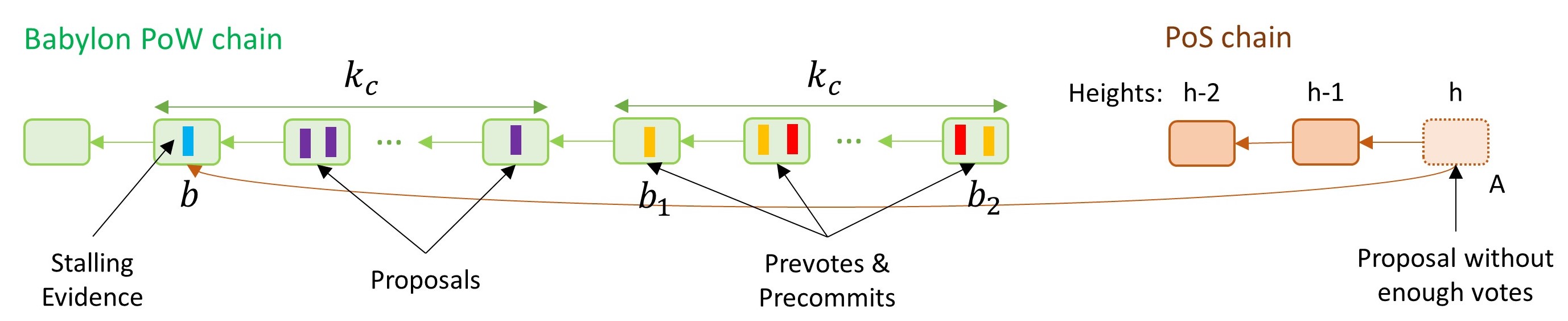}
    \caption{Slashing for stalling. A stalling evidence for height $h$, where validators failed to finalize any \pos block, is submitted to \chain and captured by the \bpow block $b$. 
    Upon observing $b$ with a stalling evidence, validators enter a new Tendermint round whose messages are recorded on \chain.
    During this round, they submit their proposals, prevotes and precommits to \chain.
    Validators, whose proposals do not appear on \chain between $b$ and $b_1$, are deemed to be unresponsive and slashed for stalling the protocol.
    Similarly, validators, whose votes for the proposal selected from the interval $[b,b_1)$ are missing between $b_1$ and $b_2$, are slashed.
    Here, the $k_c$ block intervals between $b$, $b_1$ and $b_2$ ensure that the proposals, prevotes and precommits submitted on time by the honest validators appear on \chain in the appropriate interval, thus preventing honest validators from getting slashed for stalling.}
    \label{fig:stalling}
\end{figure*}

\label{sec:protocol-description}

In this section, we specify how to obtain slashable security for any accountably-safe \pos protocol using \chainnosp-specific add-ons.
Unless stated otherwise, the accountably-safe \pos protocol is treated as a black-box which takes \pos transactions as its input and outputs a chain of finalized \pos blocks containing these transactions.
We assume that the consensus-related messages required to verify finalization of \pos blocks can be accessed by viewing the contents of the child blocks.

For concreteness, sections below focus on the interaction between the \bpow chain and a single \pos protocol.

\subsection{Handling of Commitments by \bpow}
\label{sec:construction-of-the-bpow-chain}

\pos nodes timestamp messages by posting their commitments on \chainnosp.
A commitment $h$ is a succinct representation of a piece of data $D$\footnote{In a real-world implementation, commitments will also carry metadata such as a \pos chain identifier, submitter’s signature and public key. The metadata will not be validated by the \bpow miners.}. 
\bpow miners receive commitments from the \pos nodes as pairs $(\tx, D)$, where $\tx$ is a \chain transaction that contains the commitment $h$, and $D$ is the associated data. 
Upon receiving such a pair, miners validate $h$ against $D$, \ie, check if $h$ is a succinct commitment of $D$, on top of other transaction validation procedures for $\tx$. 
However, since \chain is a generic data-available timestamping service, miners do not check the syntax or semantics of the data $D$.
If the validation succeeds, miners consider the commitment $h$ valid and include $\tx$ in the next \chain block mined.
They do not include the data $D$ in the \bpow blocks. 

Whenever a miner propagates a \bpow transaction to its peers, either directly or as part of the block body, it also attaches the associated data, so that the peers receiving the transaction can also validate its availability. 
Since \pos nodes act as light clients of the \bpow chain and are connected to the peer-to-peer network of the \bpow miners, they also obtain the data broadcast by the miners.
This ensures the availability of data across all honest \pos nodes once its commitment is validated and published by the \bpow miners.

Miners merge-mine the \bpow chain following the longest chain rule 
(\cf Appendix~\ref{sec:merge-mining} for more details).
A \bpow block is said to be \emph{valid} in the view of a miner if the \bpow transactions included in the block are valid in the miner's view.

\subsection{Generation and Validation of Commitments}
\label{sec:commitments-and-DA}

There are two types of commitments: message commitments and checkpoints.
Message commitment refers to the hash of the whole message.
For example, to timestamp a list of censored transactions, a \pos node sends the hash of the list to the miners as the message commitment $h$ and the whole list as the data $D$ (Algorithm~\ref{alg.generate.commitments}).
Then, to validate the commitment, miners and \pos nodes check if the hash of the data matches the commitment (Algorithms~\ref{alg.pow.validate.commitments} and~\ref{alg.pos.validate.commitments}).

Checkpoints are commitments of finalized \pos blocks.
A single checkpoint can commit to multiple consecutive blocks from the same \pos chain.
To post a checkpoint on \chain for consecutive blocks $B_1,...,B_n$, a \pos node first extracts the transaction roots $\txroot_i$, $i=1,..,n$, from the header $B_i.\header$ of each block (Algorithm~\ref{alg.generate.commitments}).
Then, using a binding hash function $H$, it calculates the following commitment\footnote{In a real-world application, commitment also contains the header of block $B_{n+1}$ as it contains the signatures necessary to verify the finalization of block $B_n$. We omit this fact above for brevity.}:
\begin{IEEEeqnarray}{C}
\label{eq:1}
h = H(B_1.\header||...||B_n.\header||\txroot_1||...||\txroot_n).
\IEEEeqnarraynumspace
\end{IEEEeqnarray}
Finally, it sends the commitment $h$, \ie, the checkpoint, to the miners along with the data $D$ which consists of (i) the block headers $B_1.\header,..,B_n.\header$, (ii) the block bodies, and (iii) the transaction roots $\txroot_1,..,\txroot_n$ separately from the headers.

Upon receiving a checkpoint or observing one on \chain, miners and \pos nodes parse the associated data $D$ into the block headers, block bodies and transaction roots. 
Miners view the commitment as \emph{valid} if (i) expression~\eqref{eq:1} calculated using $B_1.\header,..,B_n.\header$ and $\txroot_1,..,\txroot_n$ matches the received commitment, and (ii) the roots $\txroot_1,..,\txroot_n$ commit to the transactions in the bodies of the blocks $B_1,..,B_n$ (Algorithm~\ref{alg.pow.validate.commitments}).
\pos full nodes view the commitment as valid if conditions (i) and (ii) above are satisfied, (iii) $\txroot_1,..,\txroot_n$ are the same as the transaction roots within the block headers $B_1.\header,..,B_n.\header$ and (iv) the checkpointed \pos blocks are finalized in the given order within the PoS chain in their view (Algorithm~\ref{alg.pos.validate.commitments}).  
Although each header already contains the respective transaction root, a \bpow miner does not necessarily know the header structure of different \pos protocols.
Thus, miners receive transaction roots separately besides the block headers and bodies.
Note that miners cannot check if the transaction root $\txroot_i$ it got for a block $B_i$ is the same as the root within the header $B_i.\header$.
However, honest \pos nodes can detect any discrepancy between the transaction roots in the headers and those given as part of the data $D$, and ignore incorrect commitments.

Checkpoints are designed to enable light clients towards the \pos protocol to identify the checkpointed \pos blocks when they observe a commitment on \chain.
Unlike full nodes, \pos light clients do not download bodies of \pos blocks, thus cannot check if $\txroot_i$ commits to the body of $B_i$. 
However, since they do download \pos block headers, these light clients can extract the transaction roots from the headers, calculate expression~\eqref{eq:1} and compare it against the commitment on \chain to verify its validity.
\pos light clients trust the \chain miners to check if the transaction roots $\txroot_i$ indeed commit to the data in the bodies of the checkpointed blocks.

\subsection{Checkpointing the \pos Chain}
\label{sec:construction-of-the-bpos-ledger}

Nodes send checkpoints of all finalized blocks on the PoS chain to the \chain miners every time they observe the \bpow chain grow by $k_c$ blocks\footnote{In reality, \pos nodes do not submit a commitment of all of the blocks on the \pos chain in their view. They submit commitments of only those blocks that were not captured by previous checkpoints on \chainnosp.}.
We say that a \bpow block $b$ checkpoints a \pos block $B$ in the view of a node $\client$ (at time $t$) if (i) $B$ is a finalized \& valid block in the PoS chain in $\client$'s view, and (ii) $b$ is the first block within the longest \bpow chain in $\client$'s view (at time $t$) to contain a valid checkpoint of $B$ alongside other \pos blocks.

Checkpoints that do not include information about new \pos blocks are ignored by the \pos nodes during the interpretation of the commitments on \chain.
Thus, given two consecutive checkpoints on the \bpow chain that are not ignored by a \pos node, if they do not commit to conflicting \pos blocks, then the latter one must be checkpointing new \pos blocks extending those covered by the earlier one.

\textbf{Fork-choice Rule:} (Figure \ref{fig:checkpointing}, Algorithm~\ref{alg.find.canonical.chain}) If there are no forks on the \pos chain, \ie, when there is a single chain, it is the canonical \pos chain.

If there are multiple \pos chains with conflicting finalized blocks, \ie, a safety violation, in the view of a node $\client$ at time $t$, $\client$ orders these chains by the following recency relation : 
Chain A is {\em earlier} than chain B in $\client$'s view at time $t$ if the first \pos block that is on A but not B, is checkpointed by an earlier \bpow block than the one checkpointing the first \pos block that is on B but not A, on $\client$'s canonical \bpow chain at time $t$.
If only chain A is checkpointed in this manner on $\client$'s canonical \bpow chain, then A is earlier.
If there are no \bpow blocks checkpointing \pos blocks that are exclusively on A or B, then the adversary breaks the tie for $\client$.
The canonical \pos chain $\Ledger{\client}{t}$ is taken by $\client$ to be the earliest chain in this ordering at time $t$.
Thus, \chain provides a total order across multiple chains when there is a safety violation on the \pos chains.

\subsection{Stake Withdrawals and Slashing for Safety Violations}
\label{sec:slashing-in-the-case-of-safety-violations}

Since the \pos protocol provides accountable safety, upon observing a safety violation on the \pos chain, any node can construct a \emph{fraud proof} that irrefutably identifies $n/3$ adversarial validators as protocol violators, and send it to \chain.
Fraud proof contains checkpoints for conflicting \pos blocks along with a commitment, \ie, hash, of the evidence, \eg, double-signatures, implicating the adversarial validators.
Hence, it is valid as long as the checkpoints and the commitments are valid, and serves as an irrefutable proof of protocol violation by $n/3$ adversarial validators.

{\bf Stake withdrawal:} (Figure~\ref{fig:withdrawal}, Algorithm~\ref{alg.stake.withdrawal.and.slashing}) To withdraw its stake, a validator $\validator$ first sends a special \pos transaction called the \emph{withdrawal request} to the \pos protocol.
Given $k_w$, $\validator$ is granted permission to withdraw its stake in the view of a \pos node once the node observes that
\begin{enumerate}
    \item A block $B$ on its canonical \pos chain containing the withdrawal request is checkpointed by a block $b$ on its longest \bpow chain, \ie, the longest \bpow chain in its view.
    \item There are $k_w$ blocks building on $b$ on its longest \bpow chain, where $k_w$, chosen in advance, determines the withdrawal delay.
    \item There does not exist a valid fraud proof implicating $\validator$ in the node's longest \bpow chain.
    \end{enumerate}
Once the above conditions are also satisfied in $\validator$'s view, it submits a \emph{withdrawal transaction} to the \pos protocol, including a reference to the $k_w$-th \bpow block building on $b$.
Honest nodes consider the withdrawal transaction included in a \pos block $B'$ as \emph{valid} if $B'$ extends $B$, the block with the withdrawal request, and the above conditions are satisfied in their view.

{\bf Slashing for Safety Attacks:} Stake of a validator 
becomes slashable in the view of any \pos node which observes that condition (3) above is violated. 
In this case, nodes that sent the fraud proofs on \chain can receive part of the slashed funds as reward by submitting a \emph{reward transaction} to the \pos chain.
%The reward transaction contains a reference to the \bpow block containing the first fraud proof.

\subsection{Slashing for Liveness Violations}
\label{sec:slashing-in-the-case-of-liveness-violations}

In the rest of this section, a validator or \pos node's \bpow chain, \ie, the \bpow chain in the view of a \pos node or validator, refers to the $k_c/2$-deep prefix of the longest chain in their view.
As a liveness violation can be due to either censorship, \ie, lack of chain quality, or stalling, \ie, lack of chain growth, we analyze these two cases separately:

\subsubsection{Censorship Resilience}
\label{sec:censorship-resilience}
(Figure~\ref{fig:censor}, Algorithm~\ref{alg.censorship.slashing})
\pos nodes send commitments of censored \pos transactions to \chain via \emph{censorship complaints}.
A complaint is valid in the miners's view if the commitment matches the hash of the censored transactions.

Upon observing a valid complaint on its \bpow chain, a validator includes the censored \pos transactions within the new blocks its proposes unless they have already been included in the \pos chain or are invalid with respect to the latest \pos state.
Similarly, among new \pos blocks proposed, validators vote only for those that include the censored transactions in the block's body or prefix if the transactions are valid with respect to the latest \pos state.

Suppose a censorship complaint appears within some block $b$ on a validator $p$'s \bpow chain (Figure~\ref{fig:censor}).
Let $b'$ be the first block on $p$'s \bpow chain that checkpoints a new \pos block and extends $b$ by at least $2k_c$ blocks.
Then, a \pos block $B$ is said to be \emph{censoring} in $p$'s view if (i) it is checkpointed by a block $b''$, $b' \prec b''$ in $p$'s \bpow chain, and (ii) $B$ does not include the censored transactions in neither its body nor its prefix (\cf Algorithm~\ref{alg.censorship.slashing} for a function that detects the censoring PoS blocks with respect to a censorship complaint).

\subsubsection{Slashing for Censorship Attacks} 
\label{sec:slashing-censorship-resilience}

Stake of a validator becomes slashable in an honest \pos node $p$'s view if the validator proposed or voted for a \pos block $B$ that is censoring in $p$'s view (\eg, block B in Figure~\ref{fig:censor}).

\subsubsection{Stalling Resilience} 
\label{sec:stalling-resilience}

(Figure~\ref{fig:stalling}, Algorithm~\ref{alg.stalling.slashing}) 
A node detects that the \pos protocol has stalled if no new checkpoint committing new \pos blocks appears on its \bpow chain within $2k_c$ blocks of the last checkpoint.
In this case, it sends a \emph{stalling evidence} to \chain.
Stalling evidence is labelled with the smallest height $h$ at which a \pos block has not been finalized yet and contains a checkpoint for the \pos blocks from smaller heights.
Hence, it is valid in the miners' view if the included checkpoint is valid.

Stalling evidence signals to the validators that they should hereafter publish the \pos protocol messages, previously exchanged over the network, on \chain until a new \pos block is finalized.
For instance, in the case of \chainnosp-enhanced Tendermint, a stalling evidence on \chain marks the beginning of a new \emph{round} whose proposals and votes are recorded on \chain.
Thus, upon observing the first stalling evidence that follows the last checkpoint on \chain by $2k_c$ blocks, validators stop participating in their previous rounds and enter a new, special Tendermint round for height $h$, whose messages are recorded \emph{on-\chain}.
Each of them then pretends like the next round leader and sends a proposal message to \chain for the new round.

In the rest of this section, we focus on Tendermint \cite{tendermint} as the \chainnosp-enhanced \pos protocol for the purpose of illustration.
A summary of Tendermint is given in Appendix~\ref{sec:tendermint-summary}.
The following paragraphs explain how a Tendermint round is recorded on \chain in the perspective of a validator $p$ (\cf Algorithm~\ref{alg.stalling.slashing}).
A detailed description of this can be found in Appendix~\ref{sec:details-of-stalling-resilience}

Let $b$ denote the \bpow block that contains the first stalling evidence observed by $p$ (Figure~\ref{fig:stalling}).
Let $b_1$ and $b_2$ denote the first blocks in $p$'s \bpow chain that extend $b$ by $k_c$ and $2k_c$ blocks respectively.
If a new checkpoint for a \pos block finalized at height $h$ appears between $b$ and $b_1$, $p$ stops participating in the round on-\chain and moves to the next height, resuming its communication with the other validators through the network.
Otherwise, if there are $\geq 2f+1$ \emph{non-censoring} proposal messages signed by unique validators between $b$ and $b_1$, $p$ and every other honest validator selects the message with the largest $\textsf{validRound}$ as the \emph{unique} proposal of the round.

Once $p$ decides on a proposal $B$ and observes $b_1$ in its \bpow chain, it signs and sends prevote and precommit messages for $B$ to \chain. 
Upon seeing $b_2$ in its \bpow chain, $p$ finalizes $B$ if there are more than $2f+1$ prevotes and precommits for $B$, signed by unique validators, between $b_1$ and $b_2$.
In this case, $b_2$ is designated as the \bpow block that has checkpointed block $B$ for height $h$. 
After finalizing $B$, $p$ moves to the next height, resuming its communication with other validators through the network.

\subsubsection{Slashing for Stalling Attacks} 
\label{sec:slashing-stalling-resilience}
Consider the validator $p$ and the on-\chain Tendermint round described above and suppose there is no new checkpoint for a \pos block finalized at height $h$ between $b$ and $b_1$ (Figure~\ref{fig:stalling}).
Then, if there are less than $2f+1$ uniquely signed non-censoring proposals between $b$ and $b_1$, stake of each validator with a censoring or missing proposal becomes slashable in $p$'s view. 
Similarly, if there are less than $2f+1$ uniquely signed prevotes or precommits between $b_1$ and $b_2$ for the proposal $B$ selected by $p$, stake of each validator with a missing prevote or precommit for $B$ between $b_1$ and $b_2$ becomes slashable in $p$'s view.

To enforce the slashing of the validator's stake in the case of censorship or stalling, \pos nodes can submit a reward transaction to the \pos chain, upon which they receive part of the slashed funds.

No validator is slashed by the slashing rules for censorship or stalling if there is a safety violation on the \pos chains, in which case slashing for safety (\cf Section~\ref{sec:slashing-in-the-case-of-safety-violations}) takes precedence.

\section{Scalability of the Protocol}
\label{sec:scalability-of-the-protocol}

\chain protocol above can be used by different PoS protocols simultaneously, which raises the question of how much data \bpow miners can check for availability at any given time.
To address this, we first review the three physical limits that determine the amount and speed of on-chain data generation by the PoS blockchains:
\begin{enumerate}
    \item hot storage capacity, which caps the amount of data generated before cold storage or chain snapshot have to kick in;
    \item execution throughput, which limits the data generation speed to how fast transactions and blocks can be created, validated, and executed;
    \item communication bandwidth, which limits the data generation speed to how fast transactions and blocks can be propagated throughout the P2P network.
\end{enumerate}

Since \chain does not permanently store\footnote{\chain may store the data committed in the recent \chain blocks for the synchronization between \chain nodes.} any PoS chain data, it does not have to match aggregated storage capacity of the PoS chains to provide data protection. 
Thus, there is no storage issue for \chain to scale, namely, to support many PoS protocols.

On the other hand, \chain's data processing speed, \ie the speed with which miners validate data availability, must match the total data generation speed across the PoS chains. 
Currently, the data generation speed of individual PoS chains is mostly limited by the execution rather than communication bandwidth.
As \chain only downloads the PoS data without executing it, it should also able to accommodate many PoS chains from the speed perspective.

In the unlikely case that a certain PoS protocol is only limited by the communication bandwidth and thus generates large blocks frequently, \chain could potentially apply sampling-based probabilistic data availability checks \cite{albassam2018fraud,cmt} to significantly reduce the amount of data it needs to download and process per block, which is a promising future research direction.

\section{Reference Design with Cosmos SDK}
\label{sec:cosmos}
Cosmos is a well-known open-source blockchain ecosystem that enables customizable blockchains \cite{cosmos_modules}.
It also enables inter-blockchain communications by using Cosmos Hub (ATOM) as the trust anchor. 
Therefore, Cosmos provides both the tools through its SDK and the ecosystem required to demonstrate the implementation of a \chainnosp-enhanced PoS blockchain protocol.
To this end, we first briefly review how essential Cosmos modules work together to protect the security of the Cosmos zones, \ie, its constituent PoS blockchains, and then show how \chain can enhance the security of these zones via a straightforward module extension.

\subsection{Cosmos Overview}
Cosmos encapsulates the core consensus protocol and networking in its Tendermint consensus engine, which uses Tendermint BFT underneath. 
Several interoperable modules have been built to work with this consensus engine together as a complete blockchain system. 
Each module serves a different functionality such as authorization, token transfer, staking, slashing, etc., and can be configured to meet the requirements of the application. 
Among these modules, the following are directly related to security:
\begin{itemize}
    \item \textsf{evidence} module, which enables the submission by any protocol participant, and handling of the evidences for adversarial behaviors such as double-signing and inactivity;
    \item \textsf{slashing} module, which, based on valid evidence, penalizes the adversarial validators by means such as stake slashing and excluding it from the BFT committee;
    \item \textsf{crisis} module, which suspends the blockchain in case a pre-defined catastrophic incident appears, \eg, when the sum of stakes over all the accounts exceed the total stake of the system;
    \item \textsf{gov} module, which enables on-chain blockchain governance in making decisions such as software updates and spending community funds.
\end{itemize}

We note that these modules currently are not able to handle the aforementioned attacks such as long range attacks and transaction censorship.
Moreover, in case catastrophic incidents such as chain forking appear, system cannot recover from halt by itself. 
The incident can only be resolved via human intervention, which can be either proactive or reactive: 
Under proactive human intervention, stakeholders of the system regularly agree on and publish checkpoints on the blockchain to prevent long range attacks.
Under reactive human intervention, when a forking incident happens, stakeholders get together to decide on a fork as the canonical chain. 
Since both types of interventions require stakeholder meetings, they are part of ``social consensus''.

\subsection{Enhancing Cosmos Security with \chain}

\begin{figure}[t]
    \centering
    \includegraphics[width=\linewidth]{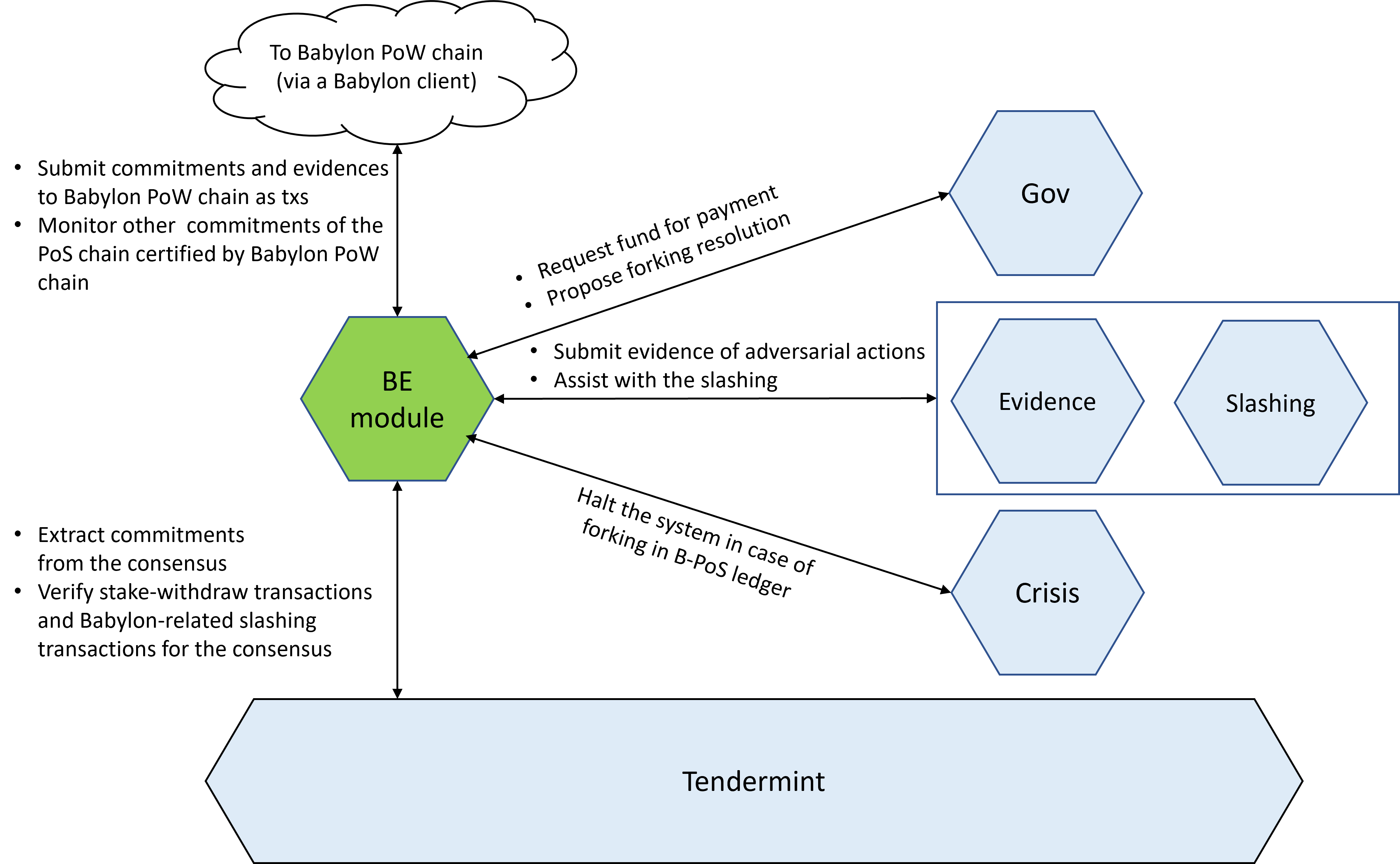}
    \caption{Enhancing Cosmos zones via a new \textsf{BE} (Babylon-enhancement) module.}
    \vspace{-0.2in}
    \label{fig:cosmos_upgrade}
\end{figure}

To enhance the security of Cosmos PoS chains with \chainnosp, we add a new module called \textsf{BE} (\chainnosp-enhancement) to the Cosmos SDK. 
This module executes the protocol described in Section~\ref{sec:protocol} and only requires straightforward interactions with existing Cosmos modules. 
Some of the key interactions are as follows (Figure~\ref{fig:cosmos_upgrade}):

\textsf{BE} implements the \chainnosp-specific add-ons such as the fork-choice rule specified in Section~\ref{sec:construction-of-the-bpos-ledger} to output the canonical PoS chain.
It monitors the PoS chain and creates the messages specified in Section~\ref{sec:protocol} such as checkpoints, fraud proofs, censorship complaints and stalling evidences. 
It communicates with the \textsf{gov} module to obtain approval for the expenditure of community funds to pay for the \bpow transaction fees.
It submits the messages mentioned above, through a customized client (Figure~\ref{fig:pos-client}) to the \bpow chain and uses \bpow transactions to pay the miners.
It also monitors the existing messages created for the same PoS chain and timestamped on \chainnosp.
In case any adversarial action is detected through the interpretation of the messages on \chain, it submits the evidences to the \textsf{evidence} module and then works with the \textsf{slashing} module to slash the adversarial validators on the PoS chain.
In the case of forking on the PoS chain, it interacts with the \textsf{crisis} module to temporarily suspend the system, and proposes resolution via the \textsf{gov} module to recover the system, where the resolution is derived using the fork-choice rule specified in Section~\ref{sec:construction-of-the-bpos-ledger}.
When withdrawal requests and \chainnosp-related PoS transactions are submitted to the Tendermint consensus engine, it helps the engine verify such transactions.

All the above interactions can be supported by existing Cosmos modules via API and data format configurations. 
These configurations are explained below:
\begin{itemize}
    \item Tendermint consensus engine: redirect the validation of stake withdrawal transactions and \bpow-related slashing transactions to the BE module.
    \item \textsf{evidence} module: add evidence types such as fraud proofs, censorship complaint and stalling evidence, corresponding to \bpow-related violations;
    \item \textsf{slashing} module: define the appropriate slashing rules as described in Sections~\ref{sec:slashing-in-the-case-of-safety-violations},~\ref{sec:slashing-censorship-resilience} and~\ref{sec:slashing-stalling-resilience};
    \item \textsf{crisis} module: add handling of safety violations reported by the \textsf{BE} module;
    \item \textsf{gov} module: add two proposal types (i) to use community funds to pay for \bpow transaction fees and (ii) to execute fork choice decision made by the \textsf{BE} module.
\end{itemize}

\section*{Acknowledgements}
We thank Joachim Neu, Lei Yang and Dionysis Zindros for several insightful discussions on this project.

\bibliographystyle{plain}
\bibliography{references}

\appendix

\section{Proofs for Section~\ref{sec:pos-security}}
\label{sec:appendix-proofs}

To formalize slashable safety and its absence thereof, we define \emph{\ssr} for the \pos protocols and state the impossibility theorem for slashable safety in the absence of additional trust assumptions:
\begin{definition}
\Ssr of a protocol is the minimum number $f$ of validators that become slashable in the view of all honest \pos nodes per Definition~\ref{def:slashable} in the event of a safety violation.
Such a protocol provides $f$-slashable-safety.
\end{definition}
\begin{theorem}
\label{thm:pos-non-slashable}
Assuming a common knowledge of the initial set of active validators, without additional trust assumptions, no \pos protocol provides both $\fS$-slashable-safety and $\fL$-$\Tconfirm$-liveness for any $\fS,\fL>0$ and $\Tconfirm<\infty$.
\end{theorem}

\begin{proof}
For the sake of contradiction, suppose there exists a \pos protocol $\PI$ that provides $\fL$-$\Tconfirm$-liveness and $\fS$-slashable-safety for some $\fL, \fS>0$ and $\Tconfirm < \infty$ without any additional trust assumptions.

Let $n$ be the number of active validators at any given time.
Let $P$, $Q'$ and $Q''$ denote disjoint sets of validators such that $P:=\{\validator_i,i=1,..,n\}$, $Q':=\{\validator'_i,i=1,..,n\}$ and $Q'':=\{\validator''_i,i=1,..,n\}$. 

Next, we consider the following two worlds, where the adversarial behavior is designated by $(\Adv,\Env)$:

\textbf{World 1:} 
$(\Adv,\Env)$ provides $P$ as the initial set of active validators.
Validators in $Q'$ are honest.
Validators in $P$ and $Q''$ are adversarial.

At time $t=0$, $\Env$ inputs transactions $\tx'_i$, $i=1,..,n$, to the validators in $P$, where $\tx'_i$ causes $\validator_i \in P$ to become \passive and $\validator'_i \in Q'$ to become active.
Validators in $P$ emulate a set of honest validators with equal size, except that they record every piece of information in their transcripts.
Since $\fL>0$ and $\Tconfirm<\infty$, there exists a constant time $T$ such that upon receiving transcripts from the set of active validators at time $T$, clients output a ledger $\LOG$ for which $\tx'_i \in \LOG$, $i=1,..,n$.
Thus, the set of active validators at time $T$ is $Q'$ in the view of any client.
As \passive validators withdraw their stake within a constant time $T'$, by time $T+T'$, all validators in $P$ have withdrawn their stake.

In parallel to the real execution above, $(\Adv, \Env)$ creates a simulated execution in its head where a different set of transactions, $\tx''_i$, $i=1,..,n$, is input to the validators in $P$ at time $t=0$.
Here, $\tx''_i$ causes $\validator_i \in P$ to become \passive and $\validator''_i \in Q''$ to become active.
Then, upon receiving the transcripts of the simulated execution at time $T$, clients would output a ledger $\LOG'$ for which $\tx_i'' \in \LOG'$, $i=1,..,n$.
Then, the set of active validators at time $T$ would be $Q''$ in the view of any client.
As \passive validators can withdraw their stake within a constant time $T'$, all validators in $P$ withdraw their stake in the simulated execution by time $T+T'$.

Finally, $(\Adv,\Env)$ spawns a \pos client $\client$ at time $T+T'$, which receives transcripts from both the simulated and real executions.
Since $\LOG$ and $\LOG'$ conflict with each other and $\fS>0$, there is a safety violation, and $\client$ identifies a set of irrefutably adversarial validators by invoking the forensic protocol, a non-empty subset of which is slashable.
As the validators in $P$ have withdrawn their stake and those in $Q'$ are honest and did not violate the protocol, this set includes at least one slashable validator from $Q''$.

\textbf{World 2:} 
World 2 is the same as World 1, except that (i) validators in $Q'$ are adversarial and those in $Q''$ are honest, and (ii) the transactions $\tx'_i$ and $\tx''_i$, $i=1,..,n$, are swapped in the description, \ie $\tx'_i$ is replaced by $\tx''_i$ and vice versa.

\begin{center}
***
\end{center}

Finally, as World 1 and 2 are indistinguishable, $\client$ again identifies a validator from $Q''$ as slashable in World 2 with probability at least $1/2$.
However, the validators in $Q''$ are honest in World 2, and could not have been identified as irrefutably adversarial, \ie contradiction.
\end{proof}

Following theorem is used for the proof of Theorem~\ref{thm:accountable-liveness}.

\begin{theorem}
\label{thm:sync-tradeoff}
For any SMR protocol that is run by $n$ validators and satisfies $\fS$-safety and $\fL$-$\Tconfirm$-liveness with $\fS, \fL > 0$ (assuming Byzantine faults) and $\Tconfirm<\infty$, it must be the case that $\fS < n - \fL$.
\end{theorem}

\begin{proof}

For the sake of contradiction, assume that there exists an SMR protocol $\PI$ that provides $\fL$-$\Tconfirm$-liveness for some $\fL>0$, $\Tconfirm<\infty$ and $\fS$-safety for $\fS = n-\fL$.
Then, the protocol should be safe when there are $n-\fL$ adversarial validators.
Let $P$, $Q$ and $R$ denote disjoint sets consisting of $\fL$, $\fL$ and $n-2\fL>0$ validators respectively, where we assume $\fL<n/2$.
Next, consider the following worlds with two clients $\client_1$ and $\client_2$ prone to omission faults, where the adversarial behavior is designated by $(\Adv,\Env)$:

\textbf{World 1:} 
$\Env$ inputs $\tx_1$ to all validators.
Those in $P$ and $R$ are honest and the validators in $Q$ are adversarial.
There is only one client $\client_1$.
Validators in $Q$ do not communicate with those in $P$ and $R$; they also do not respond to $\client_1$. 
Since $P \cup R$ has size $n-\fL$ and consists of honest validators, via $\fL$-liveness, upon receiving transcripts from the validators in $P$ and $R$, $\client_1$ outputs the ledger $[\tx_1]$ by time $\Tconfirm$.

\textbf{World 2:} 
$\Env$ inputs $\tx_2$ to all validators.
Those in $Q$ and $R$ are honest and the validators in $P$ are adversarial.
There is only one client $\client_2$.
Validators in $P$ do not communicate with those in $Q$ and $R$; they also do not respond to $\client_2$. 
Since $Q \cup R$ has size $n-\fL$ and consists of honest validators, via $\fL$-liveness, upon receiving transcripts from the validators in $Q$ and $R$, $\client_2$ outputs the ledger $[\tx_2]$ by time $\Tconfirm$.

\textbf{World 3:} 
$\Env$ inputs $\tx_1$ to the validators in $P$, $\tx_2$ to the validators in $Q$, and both transactions to the validators in $R$.
Validators in $P$ are honest, those in $Q$ and $R$ are adversarial. 
There are two clients this time, $\client_1$ and $\client_2$.
Validators in $Q$ do not send any message to any of the validators in $P$; they also do not respond to $\client_1$.
$\Env$ also omits any message sent from the validators in $P$ to $\client_2$.

Validators in $R$ perform a split-brain attack where one brain interacts with $P$ as if the input were $\tx_1$ and it is not receiving any message from $Q$ (real execution).
Simultaneously, validators in $Q$ and the other brain of $R$ start with input $\tx_2$ and communicate with each other exactly as in world 2, creating a simulated execution. 
The first brain of $R$ only responds to $\client_1$ and the second brain of $R$ only responds to $\client_2$.

Since worlds 1 and 3 are indistinguishable for $\client_1$ and the honest validators in $P$, upon receiving transcripts from the validators in $P$ and the first brain of $R$, $\client_1$ outputs $[\tx_1]$ by time $\Tconfirm$.
Similarly, since worlds 2 and 3 are indistinguishable for $\client_2$, upon receiving transcripts from the validators in $Q$ and the second brain of $R$, $\client_2$ outputs $[\tx_2]$ by time $\Tconfirm$.

Finally, there is a safety violation in world 3 since $\client_1$ and $\client_2$ output conflicting ledgers.
However, there are only $\fS=n-\fL$ adversarial validators in $Q$ and $R$, which is a contradiction.

Proof for $\fL \geq n/2$ proceeds via a similar argument, where sets $P$, $Q$ and $R$ denote disjoint sets of sizes $n-\fL$, $n-\fL$ and $2\fL-n>0$ respectively.
\end{proof}

Proof of Theorem~\ref{thm:accountable-liveness} is given below:

\begin{proof}[Proof of Theorem~\ref{thm:accountable-liveness}]
For the sake of contradiction, suppose there exists a \pos protocol $\PI$ with a static set of validators that provides $\fA$-$\Tconfirm$-accountable-liveness and $\fS$-safety for some $\fA,\fS>0$ and $\Tconfirm<\infty$ without any additional trust assumptions.
Then, there exists a forensic protocol which takes transcripts sent by the validators as input, and in the event of a liveness violation, outputs a non-empty set of validators which have irrefutably violated the protocol rules.

Let $\fL$ denote the liveness resilience of $\PI$.
By Theorem~\ref{thm:sync-tradeoff}, $\fL < n - \fS$, \ie, $\fL \leq n - 2$ as $\fS > 0$.
By definition of \alr, $\fL \geq \fA > 0$.
Let $m \geq 1$ denote the maximum integer less than $n - \fL \geq 2$ that divides $n$.
Let $P_i$, $i=1,..,n/m$ (\ie $i \in [n/m]$) denote sets of size $m$ that partition the $n$ validators into $n/m$ disjoint, equally sized groups.
We next consider the following worlds indexed by $i \in [n/m]$ where $\Env$ inputs a transaction $\tx$ to all validators at time $t=0$ and the adversarial behavior is designated by $(\Adv,\Env)$:

\textbf{World $i$:}
Validators in $P_i$ are honest.
Validators in each set $P_j$, $j \neq i, j \in [n/m]$, are adversarial and simulate the execution of $m$ honest validators in their heads without any communication with the validators in the other sets.
Validators in each $P_j$, $j \in [n/m]$ generate a set of transcripts such that upon receiving transcripts from the set of validators in $P_j$ at time $\Tconfirm$, a client outputs a (potentially empty) ledger $\LOG_j$, $j \in [n/m]$.
As $|P_j| < n - \fL$, validators in $P_i$ do not hear from the validators in $P_{j}$, $j \neq i$, and the validators in $P_j$, $j \neq i$ simulate the execution of the honest validators in world $j$ respectively, $\tx \notin \LOG_j$ for any $j \in [n/m]$.

Finally, $(\Adv,\Env)$ spawns a client at time $\Tconfirm$, which receives transcripts from both the real and the multiple simulated executions.
Since $\tx \notin \LOG_j$ for any $j \in [n/m]$, there is a liveness violation in the client's view.
As $\fA>0$, by invoking the forensic protocol with the transcripts received, client identifies a subset $S_i$ of validators as irrefutably adversarial.

\begin{center}
***
\end{center}

Finally, by definition of $S_i$, it should be the case that \linebreak
$S_i \subseteq \bigcup_{j \in [n/m], j \neq i} P_j$.
However, as worlds $i$, $i \in [n/m]$ are indistinguishable for the client, there exists a world $i^*$, $i^* \in [n/m]$, such that a node from $P_{i^*}$ is identified as adversarial in world $i^*$ with probability at least $m/n \geq 1/n$, which is non-negligible.
This is a contradiction.
\end{proof}

Note that when the number of validators $n$ is large and $\fL = n-2$, probability that the forensic protocol for accountable liveness makes a mistake and identifies an honest validator as adversarial can be small.
However, assuming that $n$ is polynomial in the security parameter of the \pos protocol, this probability will not be negligible in the security parameter.

\section{Merge Mining and Client Applications}
\label{sec:merge-mining}

\begin{figure*}[ht]
    \centering
    \includegraphics[width=0.7\linewidth]{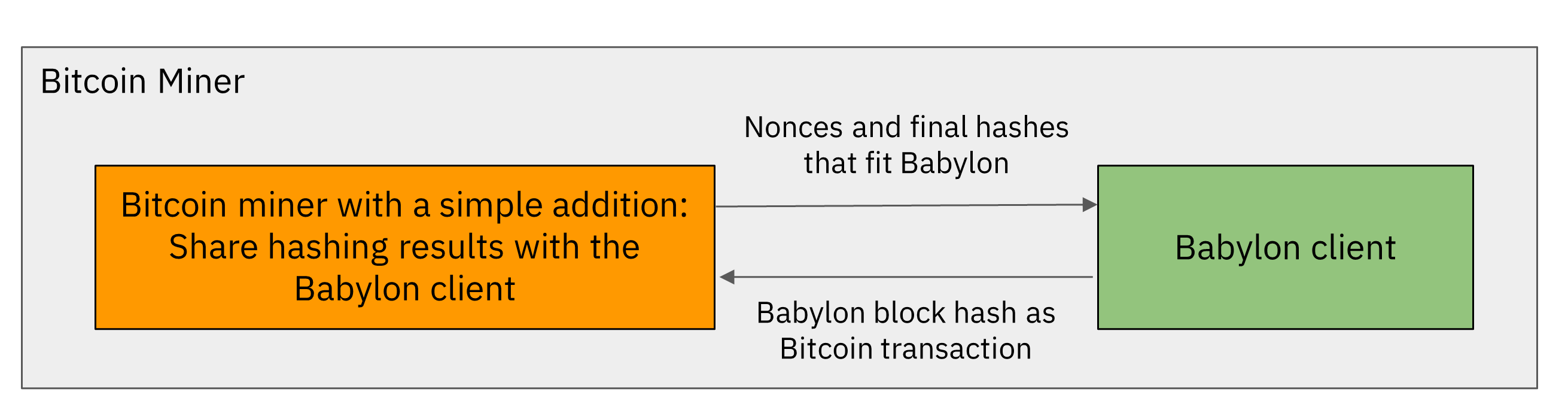}
    \caption{Interaction of the \chain client run by \bpow miners with the mining software in the context of merge-mining. \chain client uses the same hashing results generated by the Bitcoin miners as a Bitcoin client, but the criterion of mining a new \chain block based on those results is different from that of mining a Bitcoin block.}
    \label{fig:babylon-client}
\end{figure*}

\begin{figure*}[h]
    \centering
    \includegraphics[width=0.7\linewidth]{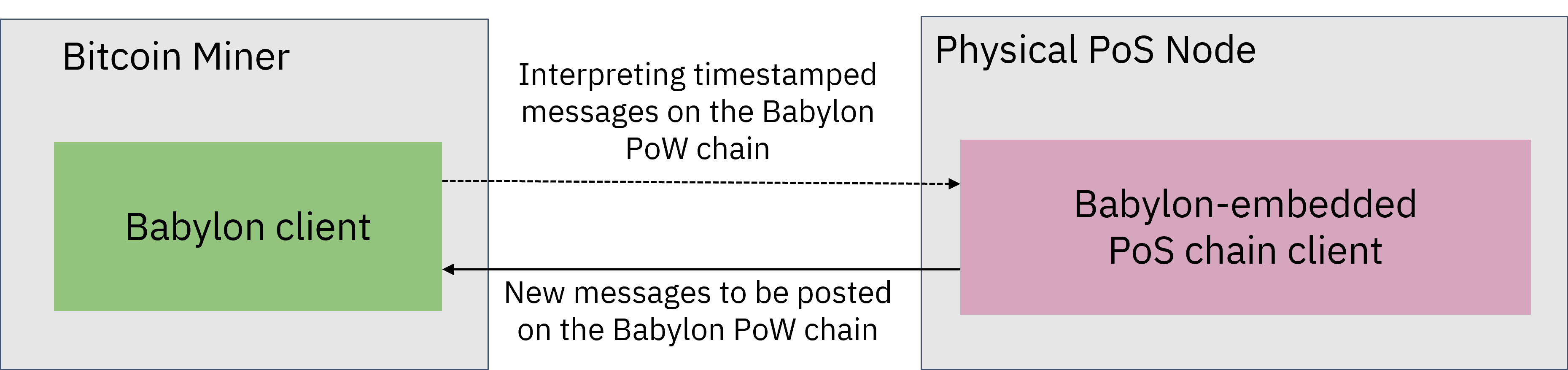}
    \caption{Interaction between the \chainnosp-embedded PoS chain client and the \chain client in the context of timestamping \pos block commitments.}
    \label{fig:pos-client}
\end{figure*}

Miners merge-mine the \bpow chain following the longest chain rule.
To merge-mine \bpow blocks, miners calculate hashes of blocks containing both Bitcoin and \bpow transactions.
Whenever a miner finds a block with its hash falling into the \emph{\chain range}, it shares the \bpow transactions in this block with its \emph{\chain client}, which extracts a \bpow block from the received contents (\cf Figure~\ref{fig:babylon-client}).
Hash of this block is then sent over Bitcoin's peer-to-peer network to be included as a Bitcoin transaction \cite{merge-mining}.
\bpow blocks have the same structure as Bitcoin blocks.
Size of the \chain range determines the chain difficulty, in turn, the growth rate $\lambda$ for the \bpow chain.

\chain client is run by the Bitcoin miners in parallel with the Bitcoin client. 
Besides exchanging nonces and hashes with the mining software for merge-mining, \chain client also records the commitments submitted by the PoS chains and checks for data availability. 
Thus, miners follow the same longest chain mining protocol as regular Bitcoin clients, except for the fact that they also check for the availability of the PoS blocks before accepting their commitments.

Similar to miners, each PoS node using \chain runs a special \emph{\chainnosp-embedded} PoS chain client (Figure~\ref{fig:pos-client}). This client is built on top of an existing PoS client, but augmented with \chainnosp-specific add-ons to allow the PoS node to post commitments and checkpoints to \chain as well as interpret the timestamps of these messages.

\section{Stalling Resilience}
\label{sec:stalling-resilience-using-babylon}

\subsection{Tendermint Summary}
\label{sec:tendermint-summary}

Tendermint consensus proceeds in heights and rounds.
Each height represents a new consensus instance and the validators cannot move on to the next height before a unique block is finalized for the previous one.
Heights consist of rounds, each with a unique leader that proposes a \pos block.
Goal of each round is to finalize a block for its height.

Rounds are divided into three steps: propose, prevote and precommit.
An honest round leader proposes a block for its round at the beginning of the propose step.
Then, during the respective steps, validators send prevote and precommit messages for the proposed block or a \emph{nil} block, depending on the proposal and their internal states.
Each honest validator maintains four variables which affect its decision whether to prevote for a proposal: $\mathsf{lockedValue}$, $\mathsf{lockedRound}$, $\mathsf{validValue}$ and $\mathsf{validRound}$.
$\mathsf{lockedValue}$ denotes the most recent non-nil block for which the validator sent a precommit message.
$\mathsf{validValue}$ denotes the most recent non-nil block for which the validator has observed $2f+1$ prevotes.
Recency of a block is determined by the round it was proposed for by the leader of that round.
Thus, $\mathsf{lockedRound}$ and $\mathsf{validRound}$ refer to the rounds for which $\mathsf{lockedValue}$ and $\mathsf{validValue}$ were proposed respectively.
At the beginning of each height, $\mathsf{lockedValue}$, $\mathsf{lockedRound}$, $\mathsf{validValue}$ and $\mathsf{validRound}$ are reset to $\bot$, $-1$, $\bot$ and $-1$ respectively.

\subsubsection{Propose}
\label{sec:propose}
If the leader of a round $r$, height $h$, is honest, it broadcasts the following proposal message at the beginning of the round if its $\mathsf{validRound} \geq 0$: $\langle \mathsf{PROPOSAL}, h, r, v=\mathsf{validValue}, vr=\mathsf{validRound} \rangle$.
Otherwise, it proposes a new valid \pos block $B$: \linebreak
$\langle \mathsf{PROPOSAL}, h, r, v=B, vr=-1 \rangle$.
Similarly, upon receiving a proposal message $\langle \mathsf{PROPOSAL}, h, r, v, vr \rangle$ (from the round leader) during the propose step of round $r$ and height $h$, an honest validator broadcasts the following prevote message $\langle \mathsf{PREVOTE}, h, r, id(v) \rangle$ for the proposal if either (i) $v$ is the same block as its $\mathsf{lockedValue}$, or (ii) $vr$ is larger than its $\mathsf{lockedRound}$.
Otherwise, it sends a prevote for a nil block: $\langle \mathsf{PREVOTE}, h, r, nil \rangle$.
Thus, by proposing its $\mathsf{validValue}$ instead of a new block when $vr \neq -1$, an honest leader ensures that honest validators locked on blocks from previous rounds will be prevoting for its proposal instead of nil blocks.

If the honest validator does not observe any proposal message within a timeout period of its entry to the propose step, it sends a prevote for a nil block.
After sending its prevote, it leaves the propose step and enters the prevote step.

\subsubsection{Prevote}
\label{sec:prevote}
Once in the prevote step, the honest validator waits until it receives $2f+1$ prevotes, for potentially different blocks, upon which it activates a prevote countdown.
If it observes $2f+1$ prevotes for a valid block $B$ proposed for round $r$ and height $h$ during this time, it sends the following precommit message $\langle \mathsf{PRECOMMIT}, h, r, id(B) \rangle$ and enters the precommit step.
It also updates its $\mathsf{lockedValue}$, $\mathsf{lockedRound}$, $\mathsf{validValue}$ and $\mathsf{validRound}$ to $B$, $r$, $B$ and $r$ respectively.
If the honest validator receives $2f+1$ prevotes for nil blocks, it sends a precommit message for a nil block: $\langle \mathsf{PRECOMMIT}, h, r, nil \rangle$.

If the honest validator does not receive $2f+1$ prevotes for a valid block $B$ before the countdown expires, it sends a precommit for a nil block.
After sending its precommit, it leaves the prevote step and enters the precommit step.

\subsubsection{Precommit}
\label{sec:precommit}
Finally, during the precommit step, our honest validator waits until it receives $2f+1$ precommit messages, for potentially different blocks, upon which it activates a precommit countdown. 
If it observes $2f+1$ precommit messages for a valid block $B$ proposed for round $r$ and height $h$, it finalizes $B$ for height $h$ and moves on to the next height $h+1$.
Otherwise, if the countdown expires or there are $2f+1$ precommit messages for nil blocks, validator enters the next round $r+1$ without finalizing any block for height $h$.

Timeout periods for proposal, prevote and precommit steps are adjusted to ensure the liveness of Tendermint under $\Delta$ synchrony when there are at least $2f+1$ honest validators.
On the other hand, the two step voting process along with the locking mechanism guarantees its safety by preventing conflicting blocks from receiving more than $2f+1$ prevotes for the same round and more than $2f+1$ precommits for the same height.

\subsection{Details of Stalling Resilience through \chain}
\label{sec:details-of-stalling-resilience}

This section presents a detailed description of how a Tendermint round is recorded on \chain and interpreted by the nodes when the \pos chain is stalled.
For the rest of this section, we assume that the \bpow chain in the view of a node refers to the $k_c/2$-deep prefix of the longest \bpow chain in its view.

To clarify the connection between censorship and stalling, we extend the definition of censoring blocks presented in Section~\ref{sec:censorship-resilience} to proposals recorded on \chain:
Consider an honest \pos node $p$ and let $b$ be a \bpow block containing a valid censorship complaint in $p$'s \bpow chain, \ie, the longest \bpow chain in $p$'s view.
Define $b'$ as the first block on $p$'s \bpow chain that contains a checkpoint and extends $b$ by at least $2k_c$ blocks.
Then, a proposal message (\cf Appendix~\ref{sec:propose}) for a block $B$ is said to be \emph{censoring} in $p$'s view if (i) the proposal was sent in response to a stalling evidence within a block $b''$ such that $b' \prec b''$ and comes after the checkpoint in $b'$ in $p$'s \bpow chain, and (ii) $B$ does not include the censored transactions neither in its body nor within its prefix.
The $2k_c$ lower bound on the gap between $b$ and $b'$ ensures that all finalized \pos blocks which exclude the censored transactions and were proposed or voted upon by honest validator are checkpointed by $b'$ or other \bpow blocks in its prefix, thus leaving no room to accuse an honest validator for censorship. 

Next, we describe a Tendermint round recorded on-\chain in the perspective of an honest validator $p$.
Suppose there is a stalling evidence for some height $h$ on $p$'s \bpow chain and the evidence is at least $2k_c$ blocks apart from the last preceding checkpoint.
Then, upon observing the first such stalling evidence recorded by a \bpow block $b$, $p$ enters a new Tendermint round for height $h$, whose messages are recorded on-\chain, and freezes the parameters $\mathsf{lockedValue}$, $\mathsf{lockedRound}$, $\mathsf{validValue}$ and $\mathsf{validRound}$ in its view.
If $p$ has observed a new valid Tendermint block become finalized at the height $h$ by that time, it sends a new checkpoint to \chain for that block.
Otherwise, $p$ signs and sends a proposal message to \chain, pretending as the leader of the new round.
Since the round is recorded on \bpow, its number $\mathsf{round}_p$ is set to a special value, $\mathsf{Babylon}$.
Thus, $p$'s $\mathsf{PROPOSAL}$ message is structured as $\langle \mathsf{PROPOSAL}, h, \mathsf{Babylon}, H(v), vr \rangle$, where either (i) $(v,vr) = (\mathsf{validValue},\mathsf{validRound})$ held by $p$ if $p$'s $\mathsf{validValue} \geq 0$, or (ii) $(v,vr) = (B,-1)$, where $B$ is a new \pos block created by $p$, if $p$'s $\mathsf{validValue} = -1$ (\cf Appendix~\ref{sec:propose}).
If $vr \geq 0$ and $p$ proposed its $\mathsf{validValue}$ as the proposal $v$, it also sends a commitment of the $2f+1$ prevote messages for $v$ to \chain along with the proposal.
This is to convince late-coming \pos nodes that the $2f+1$ prevote messages for $v$ were indeed seen by $p$ before it proposed $v$.

Note that $p$ only includes the hash of the proposed block $v$ in the proposal message unlike the proposals in Tendermint (\cf Appendix~\ref{sec:propose}).
To ensure that other \pos nodes can download $v$ if needed, miners check its availability before accepting $p$'s proposal message as valid.
Similarly, miners check the availability of prevotes upon receiving a proposal that proposes a $\mathsf{validValue}$ held by the validator.

Let $b_1$ and $b_2$ denote the first blocks on $p$'s \bpow chain that extend $b$ by $k_c$ and $2k_c$ blocks respectively.
If $p$ (or any honest \pos node) observes a checkpoint for a \pos block finalized at height $h$ between $b$ and $b_1$, it stops participating in the round recorded on-\chain and moves to the next height, resuming its communication with the other validators through the network.
Otherwise, if there are $\geq 2f+1$ non-censoring proposal messages signed by unique validators between $b$ and $b_1$, it selects the non-censoring valid block proposed with the largest $vr$ as the proposal of the `\chain round' emulated on-\chain.
If there are multiple proposals with the highest $vr$, $p$ selects the one that appears earliest between $b$ and $b_1$ on \chain.
Selecting the proposal with the largest $vr$ ensures that the honest validators can later prevote and precommit for that block without violating Tendermint rules (\cf Appendix~\ref{sec:propose}).

Once $p$ decides on a proposal $B$ and observes $b_1$ in its \bpow chain, if it is locked on \pos block, it checks if $B$ is the same block as its $\mathsf{lockedValue}$ or if the proposal's $vr$ is larger than its $\mathsf{lockedRound}$ (\cf Appendix~\ref{sec:propose}).
If so, it sends the following prevote and precommit messages for the selected proposal to \chain: $\langle \mathsf{PREVOTE}, h, \linebreak \mathsf{Babylon}, id(B) \rangle$ and $\langle \mathsf{PRECOMMIT}, h, \mathsf{Babylon}, id(B) \rangle$ (\cf Sections \ref{sec:prevote} and~\ref{sec:precommit}).
If $p$ is not locked on any \pos block, it directly sends the prevote and precommit messages. 
Unlike in Tendermint, $p$ does not wait to observe $2f+1$ prevote messages for $B$ before it sends its precommit message.
This is because the purpose of the round emulated on \chain is to catch unresponsive validators stalling the protocol, thus, does not need the two step voting.
However, it still keeps the two step voting for the purpose of consistency with the Tendermint rounds that happened off-\chainnosp.

Finally, upon observing $b_2$ in its \bpow chain, $p$ finalizes $B$ if there are more than $2f+1$ prevotes and precommits for $B$, signed by unique validators, between $b_1$ and $b_2$.
In this case, $b_2$ is designated as the \bpow block that has checkpointed the finalized block for height $h$. 
Upon finalizing $B$, $p$ moves to the next height, resuming its communication with the other validators through the network.

If $p$ observes no new checkpoints and less than $2f+1$ uniquely signed non-censoring proposals between $b$ and $b_1$, it does not send a prevote or precommit, and instead attempts to restart the round on-\chain by sending a new stalling evidence.
Similarly, if $p$ observes less than $2f+1$ uniquely signed prevotes or precommits for the selected proposal $B$ between $b_1$ and $b_2$, it does not finalize $B$, again restarting the round on-\chain.
Finally, if $p$ observes a fraud proof on \chain that implies a safety violation on the PoS chains, it stops participating in the Tendermint round on-\chain and temporarily halts finalizing new PoS blocks.

\section{Proof of Theorem~\ref{thm:main-security-informal}}
\label{sec:appendix-main-security}

We prove Theorem~\ref{thm:main-security-informal} below by showing properties \textbf{S1}-\textbf{S2} and \textbf{L1}-\textbf{L2} for the PoS chains.

\subsection{Proof of the Safety Claims \textbf{S1} \& \textbf{S2}}
\label{sec:proofs-slashable-safety}

\begin{proposition}
\label{prop:chain-quality}
If a transaction $\tx$ is sent to the miners at time $t+\Delta$ such that $|\POWLedger{\client}{t}| \leq L$ for all nodes $\client$, $\tx \in \POWLedger{\client'}{t'}$ for any honest node $\client'$, where $|\POWLedger{}{t'}| = L+k_w/2$.
\end{proposition}
Proof of Proposition~\ref{prop:chain-quality} follows from the $k_w/2$-security of \chain.

\begin{proposition}
\label{prop:conflicting-pos-blocks}
Consider a \pos block $B \in \Ledger{i}{t}$, checkpointed by a \bpow block $b \in \POWLedger{i}{t}$. 
Then, there cannot be any \bpow block in the prefix of $b$ that checkpoints a \pos block conflicting with $B$.
If $B \in \Ledger{i}{t}$ and is not checkpointed in $i$'s view by time $t$, then there cannot be any \bpow block $b' \in \POWLedger{i}{t}$ that checkpoints a \pos block conflicting with $B$.
\end{proposition}

\begin{proof}
For the sake of contradiction, suppose there exists a \bpow block $b' \preceq b$ such that $b'$ checkpoints a \pos block $B'$ that conflicts with $B$.
Then, via the fork-choice rule in Section~\ref{sec:construction-of-the-bpos-ledger}, $B \notin \Ledger{i}{t}$, \ie contradiction.
Similarly, if $B$ is not checkpointed in $i$'s view by time $t$ and there exists a \bpow block $b' \in \POWLedger{i}{t}$ such that $b'$ checkpoints a \pos block $B'$ that conflicts with $B$, again via the fork-choice rule in Section~\ref{sec:construction-of-the-bpos-ledger}, $B \notin \Ledger{i}{t}$, \ie contradiction.
\end{proof}

To show the safety claims \textbf{S1} and \textbf{S2}, we prove that if \chain is secure with parameter $k_w/2$, then whenever there is a safety violation on the PoS chains, at least $1/3$ of the validator set becomes slashable in the view of all honest \pos nodes.

\begin{proof}
Suppose there is a safety violation on the PoS chains and $\Ledger{i}{t}$ observed by an honest node $i$ at time $t$ conflicts with $\Ledger{j}{t'}$ observed by an honest node $j$ at time $t' \geq t$.
Let $B_1$ and $B_2$ denote the first two conflicting \pos blocks on $\Ledger{i}{t}$ and $\Ledger{j}{t'}$ respectively.
Via synchrony, by time $t'+\Delta$, every honest node observes $B_1$ and $B_2$, their prefixes and the protocol messages attesting to their PoS-finalization.
Since the \pos protocol has an accountable safety resilience of $1/3$, upon inspecting the blocks, their prefixes and the associated messages, any node can irrefutably identify $1/3$ of the validator set for $B_1$ and $B_2$ as having violated the protocol, and submit a fraud proof to \chain by time $t'+\Delta$.
Let $S$ denote the set of the adversarial validators witnessed by the fraud proof.

For the sake of contradiction, assume that there is a validator $\validator \in S$ that has not become slashable in the view of an honest node $\client$. 
Then, there exists a time $t_0$ and a \pos block $B'_2$ containing $\validator$'s withdrawal request such that $B'_2$ is checkpointed by a \bpow block $b'_2$ that is at least $k_w$-deep in $\POWLedger{\client}{t_0}$ and there is no fraud proof showing $\validator$'s misbehavior on $\POWLedger{\client}{t_0}$ (\cf Section~\ref{sec:slashing-in-the-case-of-safety-violations}).
Now, suppose $b'_2$ has not become at least $k_w/2$ deep in the longest \bpow chain in the view of any node, including adversarial ones, by time $t'+\Delta$.
In this case, since the fraud proof submitted to the \bpow chain by time $t'+\Delta$ will appear and stay in the canonical \bpow chain of all honest nodes within $k_w/2$ block-time of $t'$ by Proposition~\ref{prop:chain-quality}, fraud proof will be on $\POWLedger{\client}{t_0}$ as well.
However, this is a contradiction, implying that there must be at least one, potentially adversarial, node $j'$, which observes $b'_2$ become $k_w/2$ deep in its longest \bpow chain at some time $s \leq t'+\Delta$.

Next, we analyze the following cases:
\begin{itemize}
    \item \textbf{Case 1:} There exists a \bpow block $b_1 \in \POWLedger{i}{t}$ such that $b_1$ checkpoints $B_1$ and $b_1 \preceq b'_2 \in  \POWLedger{i}{t}$.
    \item \textbf{Case 2:} There exists a \bpow block $b_1 \in \POWLedger{i}{t}$ such that $b_1$ checkpoints $B_1$ and $b'_2 \prec b_1$.
    \item \textbf{Case 3:} $b'_2 \notin \POWLedger{i}{t}$.
    \item \textbf{Case 4:} There does not exist a \bpow block $b_1 \in \POWLedger{i}{t}$ checkpointing $B_1$ at time $t$ and $b'_2 \in \POWLedger{i}{t}$.
\end{itemize}

\textbf{Case 1:} $b_1 \preceq b'_2$. 
By Proposition~\ref{prop:conflicting-pos-blocks}, $b_1 \notin \POWLedger{j}{t'}$, which implies that the $k_w/2$ blocks building on $b'_2$ in $\POWLedger{j'}{s}$ are not in $\POWLedger{j}{t'}$ at time $t'\geq s-\Delta$.
However, this is a contradiction with the $k_w/2$-safety of \chain.

\textbf{Case 2:} $b'_2 \prec b_1$. 
If $B'_2$ conflicts with $B_1$, as $B_1 \in \Ledger{i}{t}$, via Proposition~\ref{prop:conflicting-pos-blocks}, $b'_2 \prec b_1$ cannot be true, \ie contradiction.
On the other hand, if $b'_2 \prec b_1$ and $B'_2$ does not conflict with $B_1$, then $B'_2 \prec B_1$, in which case $\validator$ cannot be in the validator set $S$ that voted for $B_1$, \ie contradiction.
%This is again a contradiction as $\validator \in S$ by assumption.

\textbf{Case 3:} $b'_2 \notin \POWLedger{i}{t}$.
Suppose $t \geq s$.
Since $\POWLedger{i}{t}$ does not contain the $k_w/2$ \bpow blocks following $b'_2$ in $j'$'s canonical \bpow chain at time $s$, in this case, \chain cannot be safe with parameter $k_w/2$, \ie, contradiction. 

On the other hand, if $t < s$, we consider the following sub-cases:
\begin{itemize}
    \item \textbf{Case 3-a:} There exists a \bpow block $b_2 \in \POWLedger{j}{t'}$ such that $b_2$ checkpoints $B_2$ and $b_2 \preceq b'_2$.
    \item \textbf{Case 3-b:} There exists a \bpow block $b_2 \in \POWLedger{j}{t'}$ such that $b_2$ checkpoints $B_2$ and $b'_2 \prec b_2$.
    \item \textbf{Case 3-c:} $b'_2 \notin \POWLedger{j}{t'}$.
    \item \textbf{Case 3-d:} There does not exist a \bpow block $b_2 \in \POWLedger{j}{t'}$ checkpointing $B_2$ at time $t'$ and $b'_2 \in \POWLedger{j}{t'}$.
\end{itemize}

\item \textbf{Case 3-a:} $b_2 \preceq b'_2$.
In this case, as $t < s$, by time $s+\Delta$, node $i$ would have observed both \pos blocks $B_1$ and $B_2$ along with their prefixes and sent a fraud proof to the \bpow miners. 
Then, by Proposition~\ref{prop:chain-quality}, the fraud proof will appear in the prefix of the $k_w$-th \bpow block building on $b'_2$ in $\client$'s canonical \bpow chain.
However, this is a contradiction with the assumption that $\validator$ has withdrawn its stake in $\client$'s view. 

\textbf{Case 3-b:} $b'_2 \prec b_2$. 
If $B'_2$ conflicts with $B_2$, as $B_2 \in \Ledger{j}{t'}$, via Proposition~\ref{prop:conflicting-pos-blocks}, $b'_2 \prec b_2$ cannot be true, \ie contradiction.
On the other hand, if $b'_2 \prec b_2$ and $B'_2$ does not conflict with $B_2$, then $B'_2 \prec B_2$, in which case $\validator$ cannot be in the validator set $S$ that voted for $B_2$, again a contradiction.

\textbf{Case 3-c:} $b'_2 \notin \POWLedger{j}{t'}$.
In this case, $\POWLedger{j}{t'}$ does not contain the $k_w/2$ \bpow blocks following $b'_2$ in $j'$'s canonical \bpow chain at time $s \leq t'+\Delta$.
However, this contradicts with the $k_w/2$-safety of the \bpow chain.

\textbf{Case 3-d:} There does not exist a \bpow block $b_2 \in \POWLedger{j}{t'}$ checkpointing $B_2$ at time $t'$ and $b'_2 \in \POWLedger{j}{t'}$.
In this case, if $B'_2$ conflicts with $B_2$ and $B_2 \in \Ledger{j}{t'}$, via Proposition~\ref{prop:conflicting-pos-blocks}, $b'_2 \in \POWLedger{j}{t'}$ cannot be true, \ie contradiction.
On the other hand, if $B'_2 \prec B_2$, $\validator$ cannot be in the validator set $S$ that voted for $B_2$, again a contradiction.
Finally, if $B_2 \preceq B'_2$, then $b'_2$ also checkpoints $B_2$ by the monotonicity of checkpoints (\cf Section~\ref{sec:construction-of-the-bpos-ledger}), which is a contradiction with the assumption that there does not exist a \bpow block $b_2 \in \POWLedger{j}{t'}$ checkpointing $B_2$ at time $t'$.

\textbf{Case 4:} There does not exist a \bpow block $b_1 \in \POWLedger{i}{t}$ checkpointing $B_1$ at time $t$ and $b'_2 \in \POWLedger{i}{t}$.
In this case, if $B'_2$ conflicts with $B_1$, as $B_1 \in \Ledger{i}{t}$, via Proposition~\ref{prop:conflicting-pos-blocks}, $b'_2 \in \POWLedger{i}{t}$ cannot be true, implying contradiction.
On the other hand, if $B'_2 \prec B_1$, $\validator$ cannot be in the validator set $S$ that voted for $B_1$, again a contradiction.
Finally, if $B_1 \preceq B'_2$, then $b'_2$ also checkpoints $B_1$ by the monotonicity of checkpoints (\cf Section~\ref{sec:construction-of-the-bpos-ledger}), which is a contradiction with the assumption that there does not exist a \bpow block $b_1 \in \POWLedger{i}{t}$ checkpointing $B_1$ at time $t$.

Thus, by contradiction, we have shown that if the \chain chain satisfies $k_w/2$-security, none of the validators in the set $S$ can withdraw their stake in the view of any honest node by time $t'+\Delta$.
Moreover, since they have been irrefutably identified as protocol violators by every honest node by time $t'+\Delta$, they are slashable per Definition~\ref{def:slashable}.
Consequently, whenever there is a safety violation on the PoS chains, at least $1/3$ of the validators become slashable if \chain satisfies $k_w/2$-security.
\end{proof}

No honest validator becomes slashable in the view of any \pos node due to a safety violation since fraud proofs never identify an honest validator as a protocol violators via the accountability guarantee provided by Tendermint \cite{tendermint_thesis}. Thus, even if the adversary compromises the security of the \bpow chain, it cannot cause honest validators to get slashed for a safety violation on the \pos chain.
However, when the security of the \bpow chain is violated, honest validators might be subjected to slashing for censorship and stalling as the arrow of time determined by \chain is distorted.
Since accountable liveness is impossible without external trust assumptions, slashing of honest nodes is unavoidable if the adversary can cause arbitrary reorganizations of blocks on \chain.
This point is addressed in the next section.

\subsection{Proof of the Liveness Claims \textbf{L1} \& \textbf{L2}}
\label{sec:proofs-slashable-liveness}

In this section, we prove that if \chain is secure with parameter $k_c/2 \leq k_w/2$, whenever there is a liveness violation exceeding $\Theta(k_c)$ block-time, at least $1/3$ of the active validator set becomes slashable in the view of all honest nodes.
Let $f$ denote the safety and liveness resilience of Tendermint such that $n=3f+1$.
In the rest of this section, we assume that there is no safety violation on the \pos chains and there is no fraud proof posted on \chain that accuses $f+1$ active validators of equivocating on prevote or precommit messages.
We will relax this assumption and consider the interaction between safety and liveness violations at the end of this section.

In the rest of this section, we will assume that the \chain chain is $k_c/2$-secure per Definition~\ref{def:pow-security} unless stated otherwise.
Moreover, in this section, \bpow chains in the view of honest nodes or validators will refer to the $k_c/2$-deep prefix of the longest chain in their view.
Under the $k_c/2$-security assumption for \chain, this reference ensures that (i) the \bpow chains observed by different honest nodes at any given time are prefixes of each other, and (ii) the data behind every commitment appearing on the \bpow chains held by the honest nodes is available per Definition~\ref{def:pow-security}.

Let $\POWLedger{}{t}$, without any node specified, denote the shortest \bpow chain in the view of the honest nodes at time $t$.
Thus, $\POWLedger{}{t}$ is a prefix of all \bpow chains held by the honest nodes at time $t$.
Moreover, by the synchrony assumption, we deduce that all \bpow chains held by the honest nodes at time $t-\Delta$ are prefixes of $\POWLedger{}{t}$.
Using the notation $\POWLedger{}{t}$, we can state the liveness property of the \bpow chain in the following way:
\begin{proposition}
\label{prop:chain-growth}
If a transaction $\tx$ is sent to the miners at time $t+\Delta$ such that $|\POWLedger{}{t}| = L$, $\tx \in  \POWLedger{}{t'}$, where $|\POWLedger{}{t'}| = L+k_c$.
\end{proposition}
Proof of Proposition~\ref{prop:chain-growth} follows from the $k_c/2$-security of \chain.

As stated in Section~\ref{sec:slashing-in-the-case-of-liveness-violations}, at least one honest node sends checkpoints for new finalized \pos blocks every time it observes the \bpow chain grow by $k_c$ blocks.
Moreover, if a \pos block is finalized for the first time in an honest node's view at time $t$, it is finalized in every honest node's view by time $t+\Delta$ via the synchrony assumption.
\begin{proposition}
\label{prop:new-checkpoint}
Suppose a \pos block is finalized in the view of an honest node at time $t$ such that $|\POWLedger{}{t}| = L$.
Then, a checkpoint for the finalized block appears in $\POWLedger{}{t'}$, where $|\POWLedger{}{t'}| = L+2k_c$.
\end{proposition}
Proof of Proposition~\ref{prop:new-checkpoint} follows from Proposition~\ref{prop:chain-growth} and the assumption above.

\begin{lemma}
\label{lem:new-liveness-event}
Suppose $|\POWLedger{}{t}| = L$ and $|\POWLedger{}{t'}| = L+5k_c$ for times $t$, $t'$, and the height of the last \pos block finalized in any honest node's view by time $t$ is $h-1$.
Then, either a new non-censoring Tendermint block for height $h$ is checkpointed within the interval $(t,t']$ or $f+1$ active validators must have violated the slashing rules for censorship and stalling in Sections~\ref{sec:slashing-censorship-resilience},~\ref{sec:details-of-stalling-resilience} and~\ref{sec:slashing-stalling-resilience} in the view of all honest nodes.
\end{lemma}

\begin{proof}
Let $t_i$ denote the first time such that $|\POWLedger{}{t_i}| = L+ik_c$.
If no new checkpoint for height $h$ appears on the \bpow chain by time $t_2$, at least one honest node must have sent a stalling evidence to \chain by that time.
Thus, via Proposition~\ref{prop:chain-growth}, by time $t_3$, there exists a stalling evidence recorded in a block $b$ in the \bpow chains held by all honest nodes.
Similarly, by time $t_4$, block $b_1$ that extends $b$ by $k_c$ blocks appears in the \bpow chain of all honest nodes.
At this point, there are three possibilities:
\begin{enumerate}
    \item There is a new checkpoint in the prefix of $b_1$ that commits to a new \pos block finalized for height $h$.
    \item There are less than $2f+1$ proposal messages signed by unique validators, proposing non-censoring \pos blocks between $b$ and $b_1$.
    \item There are $2f+1$ or more proposal messages signed by unique validators, proposing non-censoring \pos blocks between $b$ and $b_1$.
\end{enumerate}
Case 1 implies that there is a new checkpoint in every honest node's view by time $t'$ and case 2 implies that more than $f+1$ active validators must have violated the slashing rules in Sections~\ref{sec:slashing-censorship-resilience} and~\ref{sec:slashing-stalling-resilience} in the view of all honest nodes.

If case 3 happens, validators are required to send prevote and precommit messages for the unique non-censoring \pos block $B$ selected as described in Section~\ref{sec:details-of-stalling-resilience} by time $t_4$.
Moreover, by time $t_5$, block $b_2$ that extends $b_1$ by $k_c$ blocks appears in the \bpow chain of all honest nodes.
Thus, if case 3 happens, there are two possibilities at time $t_5$:
\begin{enumerate}
    \item There are less than $2f+1$ prevote or precommit messages for $B$, signed by unique validators, between $b_1$ and $b_2$.
    \item There are $2f+1$ or more prevote and precommit messages for $B$, signed by unique validators, between $b_1$ and $b_2$.
\end{enumerate}
In the former case, more than $f+1$ active validators must have violated the slashing rules in Section~\ref{sec:slashing-stalling-resilience} in the view of all honest nodes.
In the latter case, block $B$ is PoS-finalized and block $b_2$ acts as a new checkpoint for the \pos block at height $h$ as stated in Section~\ref{sec:details-of-stalling-resilience}.
\end{proof}

\begin{lemma}
\label{lem:no-honest-slashing}
No honest validator violates the slashing rules for censorship and stalling in Sections~\ref{sec:slashing-censorship-resilience},~\ref{sec:details-of-stalling-resilience} and~\ref{sec:slashing-stalling-resilience} in the view of any honest node.
\end{lemma}

\begin{proof}
\textbf{Censorship:} 
Suppose a \bpow block $b$ containing a censorship complaint first appears in the \bpow chains of \emph{all} honest nodes at time $t_0$, $|\POWLedger{}{t_0}|=L$. 
Let $t_i$ denote the first time such that $|\POWLedger{}{t_i}| = L+ik_c$.
For the sake of contradiction, suppose an honest validator proposed or voted (precommit or commit) for a censoring \pos block $B$.
Let $b'$ denote the \bpow block containing the first checkpoint after $b$ that is \emph{more than} $2k_c$ blocks apart from $b$.
Given $b$ and $b'$, there are two cases for $B$:
\begin{itemize}
    \item $B$ is PoS-finalized and checkpointed by a \bpow block $b''$ such that $b' \prec b''$ in all of the \bpow chains held by the honest nodes. 
    \item A proposal for $B$ appears in response to a stalling evidence recorded by a \bpow block $b''$ such that $b' \prec b''$ (\cf Section~\ref{sec:details-of-stalling-resilience}).
\end{itemize}

Let $h$ denote the largest height for which a \pos block was finalized in the view of any honest node by time $t_0$.
Then, by Proposition~\ref{prop:new-checkpoint}, a checkpoint for that block appears in the \bpow chains held by all honest nodes by time $t_2$ in the prefix of the $2k_c$-th block extending $b$.
Thus, by the monotonicity of checkpoints, the checkpoint within block $b'$ must cover the \pos block finalized for height $h+1$.
Moreover, since honest validators cannot propose or vote for \pos blocks at heights $\geq h+2$ before a block for height $h+1$ is finalized, if an honest validator proposed or voted for a \pos block excluding the censored transactions by time $t_0$, by definition, this block must have been proposed for height $h+1$ or lower.
Consequently, as honest validators would not propose or vote for new \pos blocks excluding the censored transactions after time $t_0$, such a block can only be finalized at heights $\leq h+1$ and its checkpoint can only appear as a checkpoint either in block $b'$ or within some other \bpow block in its prefix.
Hence, given the first case above, it is not possible for an honest validator to have proposed or voted for $B$, \ie contradiction.

In the second case, by Proposition~\ref{prop:chain-growth}, any stalling evidence recorded by a \bpow block $b''$ such that $b' \prec b''$, must have been sent after block $b$ was observed by every honest node on their \bpow chains.
From the explanation above, we know that a block that excludes the censored transactions and voted upon or proposed by honest validators can only be finalized at heights $\leq h+1$ and its checkpoint can only appear as a checkpoint either in block $b'$ or within some other \bpow block in its prefix.
Thus, if a proposal for block $B$ appears in response to a stalling evidence recorded by $b'$ or one of its descendants, $B$ must be for a height larger than $h+1$.
Since no honest validator proposes or votes for \pos blocks from heights larger than $h+1$ that does not contain the censored transactions, it is not possible for an honest validator to have proposed $B$, \ie contradiction.
Consequently, no honest validator could have committed a slashable offense per the rules in Sections~\ref{sec:slashing-censorship-resilience} and~\ref{sec:details-of-stalling-resilience}, thus become slashable, in the view of any honest node.

\textbf{Stalling:} 
Consider a stalling evidence for height $h$ recorded by a \bpow block $b$ that is first observed by every honest node at time $t_0$ such that $|\POWLedger{}{t_0}|=L$. 
Define $b_1$ and $b_2$ as the \bpow blocks that extend $b$ by $k_c$ and $2k_c$ blocks respectively in the \bpow chains of all honest nodes.
For the sake of contradiction, suppose there is no checkpoint between $b$ and $b_1$ for height $h$ and a proposal message by an honest validator is missing in the same interval.
In this case, the validator could not have seen a \pos block finalized for height $h$ by time $t_0$ as it would have otherwise sent a checkpoint for it and the checkpoint would have been recorded between $b$ and $b_1$ by Proposition~\ref{prop:chain-growth}.
Hence, the validator must have sent a proposal message by time $t_0$, which would appear between $b$ and $b_1$ by Proposition~\ref{prop:chain-growth}.

Next, assume that there are at least $2f+1$ proposal messages for non-censoring blocks between $b$ and $b_1$, and no checkpoint for height $h$.
For the sake of contradiction, suppose that prevote or precommit messages by an honest validator $p$ are either missing between the blocks $b_1$ and $b_2$ or $p$ sent these messages for a \pos block that is different from the proposal selected by another honest node $p'$.
In either case, $p$ becomes slashable for stalling in the view of $p'$. 
To rule out both cases, we first show that all honest validators choose the same block as their proposal, under our initial assumption that there is no fraud proof accusing $f+1$ active validators.

If all of the proposal messages between $b$ and $b_1$ for non-censoring \pos blocks have $vr = -1$, then all honest validators choose the block proposed by the earliest message between $b$ and $b_1$ as the proposal of the on-\chain Tendermint round (\cf Section~\ref{sec:details-of-stalling-resilience}).
If there is only one proposal message between $b$ and $b_1$ with $vr \neq -1$ for a non-censoring \pos block, then all honest validators choose this block as the proposal for the on-\chain Tendermint round.
Finally, suppose there are at least two proposal messages between $b$ and $b_1$ with $vr \neq -1$ and for different non-censoring \pos blocks.
Then, different \pos blocks must have acquired at least $2f+1$ prevote messages for the same round $vr \neq -1$ of the latest height $h$ (\cf Section~\ref{sec:tendermint-summary}).
For different \pos blocks to acquire $2f+1$ prevote messages for the same round $vr$ of height $h$, at least $f+1$ validators from the active validator set must have sent prevotes for conflicting blocks proposed for the same round $vr$.
In this case, a fraud proof implicating these $f+1$ current validators will be created by an honest node and will eventually appear on the \bpow chain, which contradicts with our assumption on fraud proofs. 
Thus, under this assumption, either the maximum $vr$ among all proposal messages is greater than $-1$, in which case the messages with the largest $vr$ propose the same non-censoring \pos block, or all of them have $vr=-1$, in which case the non-censoring block within the earliest proposal on the \bpow chain between $b$ and $b_1$ is selected by the nodes.
Hence, in any of the cases, $p$ could not have chosen, as its proposal for the on-\chain round, a \pos block that is different from the proposal selected by another honest node $p'$. 

Finally, if $p$'s $\mathsf{lockedRound} = -1$, then it directly sends prevote and precommit messages for the selected proposal upon observing block $b_1$ (\cf Section~\ref{sec:details-of-stalling-resilience}).
Otherwise, if its $\mathsf{lockedRound} \neq -1$, its $-1 \neq \mathsf{validRound} \geq \mathsf{lockedRound}$, and, as an honest validator, $p$ must have sent a proposal message with $vr_{\text{sent}}$ equal to its $\mathsf{validRound}$.
Then, for the $vr_{\text{sel}}$ of the selected proposal, which is the maximum among the $vr$ values of all the proposals between $b$ and $b_1$, thus $vr_{\text{sel}} \geq vr_{\text{sent}} = \mathsf{lockedRound} \geq \mathsf{lockedRound}$ of $p$, there are two cases:
\begin{itemize}
    \item $vr_{\text{sel}}$ exceeds $p$'s $\mathsf{lockedRound}$, in which case $p$ sends prevote and precommit messages for the proposal upon seeing block $b_1$.
    \item $vr_{\text{sel}} = \mathsf{lockedRound} = \mathsf{validRound} = vr_{\text{sent}} \neq -1$.
    In this case, $p$ has sent a proposal message with the largest $vr$ and must have proposed the same \pos block as the one suggested by the selected proposal message via the reasoning above, which is $\mathsf{validValue}$ by Tendermint rules (\cf Section~\ref{sec:propose}).
    As $p$'s $\mathsf{lockedRound} = \mathsf{validRound}$, it should be the case that $p$'s $\mathsf{lockedValue} = \mathsf{validValue}$ unless again there are $f+1$ active validators that sent prevote messages for conflicting blocks at round $\mathsf{lockedRound}$.
    Thus, $p$ again sends prevote and precommit messages for the selected proposal upon seeing block $b_1$.  
\end{itemize}
Thus, $p$ sends prevote and precommit messages upon seeing $b_1$ for the selected proposal which is the same across all honest nodes, and these votes appear on \chain in the view of all honest nodes between blocks $b_1$ and $b_2$ by Proposition~\ref{prop:chain-growth}. 
Consequently, $p$ could not have committed a slashable offense per the rules in Section~\ref{sec:slashing-stalling-resilience}, and thus become slashable, in the view of any node $p'$, \ie, contradiction.
\end{proof}

Liveness part of Theorem~\ref{thm:main-security-informal} is proven below:

\begin{proof}
Let $t_i$ denote the first time such that $|\POWLedger{}{t_i}| = L+ik_c$.
Suppose a censorship complaint is sent to \chain at time $t_0$.
Then, by Proposition~\ref{prop:chain-growth}, the complaint appears in the \bpow chain of every honest node within some block $b$ by time $t_1$.
Let $h-1$ be the height of the last \pos block finalized in any honest node's view by time $t_3$.

Suppose no more than $f$ active validators violate the slashing rules in Sections~\ref{sec:slashing-censorship-resilience},~\ref{sec:details-of-stalling-resilience} and~\ref{sec:details-of-stalling-resilience} in the view of all honest nodes.
Then, via Lemma~\ref{lem:new-liveness-event}, a new non-censoring \pos block for height $h$ is checkpointed by some \bpow block $b'$ within the interval $(t_3,t_8]$.
Similarly, again via Lemma~\ref{lem:new-liveness-event}, a new non-censoring \pos block $B$ for a height larger than $h$ is checkpointed by some \bpow block $b''$ by the time $t_{13}$.
As $B$ is non-censoring and appears within a \bpow block $b''$ such that $b' \prec b''$ (\cf Sections~\ref{sec:censorship-resilience} and~\ref{sec:details-of-stalling-resilience}), by definition of censoring blocks, it includes the censored transactions.
Consequently, unless $f+1$ active validators violate the slashing rules for censorship and stalling in the view of all honest nodes, any valid \pos transaction becomes finalized and checkpointed within $13k_c$ blocktime of the time censorship is detected.

Finally, via Lemma~\ref{lem:no-honest-slashing}, no honest validator violates the slashing rules for censorship and stalling in Sections~\ref{sec:slashing-censorship-resilience},~\ref{sec:details-of-stalling-resilience} and~\ref{sec:slashing-stalling-resilience} in the view of any honest node.
Hence, if there is a liveness violation for a duration of more than $13k_c$ block-time, either of the following conditions must be true:
\begin{itemize}
    \item \textbf{L1:} More than $f+1 \geq n/3$ active validators, all of which are adversarial, must have violated the slashing rules for censorship and stalling in the view of all honest nodes.
    Thus, these protocol violators will be identified as irrefutably adversarial in the view of all honest nodes. 
    As they are active validators and have not withdrawn their stake, they also become slashable in the view of all honest nodes.
    \item \textbf{L2:} The \bpow chain is not secure with parameter $k_c/2$.
\end{itemize}
This concludes the proof of the liveness claims in Theorem~\ref{thm:main-security-informal} under the assumption that no fraud proof appears on \chain accusing $f+1$ validators of a slashable offense for safety.

Finally, we relax the assumption on the fraud proof and safety.
Suppose there is a fraud proof on \chain implicating $f+1$ active validators in a safety violation on the \pos chains.
Then, the \pos protocol is temporarily halted and no validator can be slashed for any slashing rule other than for a safety violation (\cf Section~\ref{sec:slashing-in-the-case-of-safety-violations}).
This precaution prevents adversarial validators from making honest ones slashable due to censorship or stalling in the event of a safety violation on the \pos chains; however, results in a liveness violation as the \pos protocol stops finalizing new blocks temporarily.
Note that if the \bpow chain is secure with parameter $k_c/2$, by the assumption $k_c \leq k_w$, it is also secure with parameter $k_w$.
Hence, if the protocol halts due to a safety violation on the \pos chains, at least one of the following conditions should be true:
\begin{itemize}
    \item \bpow chain is not secure with parameter $k_c/2$, implying clause \textbf{L2}.
    \item \bpow chain is secure with parameters $k_c$ and $k_w$, implying that at least $f+1 > n/3$ adversarial validators become slashable in the view of all honest nodes via the proof Section~\ref{sec:proofs-slashable-safety}, \ie, clause \textbf{L1}.
\end{itemize}
Thus, even though the halting of the protocol due to a safety violation could cause a liveness violation, the liveness claims of Theorem~\ref{thm:main-security-informal} hold in this case as well.
This concludes the liveness proof.
\end{proof}

\begin{algorithm*}
  \caption{\label{alg.pow.validate.commitments} Validation of the commitments sent to \chain by the miners. Returns true if the commitment is valid given the data $D$.}
  \begin{algorithmic}[3]
    \Function{\sf pow\_validate\_commitment}{\sf commitment, data, type}
    
        \If{$\textsf{type} == \textsf{checkpoint}$}
            \CommentLine{Parse the data for the checkpoint.}
            \Let{\textsf{block\_headers}, \textsf{transaction\_roots}, \textsf{block\_bodies} }{\textsf{parse\_block\_data}(\textsf{data})}
            
            \CommentLine{Check if the transaction roots sent as part of the data commit to the PoS block bodies.}
            \For{$(\textsf{block\_header, transaction\_root, block\_body}) \gets \textsf{block\_headers}, \textsf{transaction\_roots}, \textsf{block\_bodies}$}       
                \If{$\textsf{transaction\_root} \neq H(\textsf{block\_body})$}
                    \State\Return $\textsf{False}$
                \EndIf
            \EndFor
            \CommentLine{Check if the checkpoint was correctly calculated.}
            \If{$\textsf{commitment} == H(\textsf{block\_headers}\ ||\ \textsf{transaction\_roots})$}
                \State\Return $\textsf{True}$
            \EndIf
            \State\Return $\textsf{False}$
        \ElsIf{$\textsf{type} == \textsf{message}$}
            \CommentLine{Check if the message commitment was correctly calculated.}
            \If{$\textsf{commitment} == H(\textsf{data})$}
                \State\Return $\textsf{True}$
            \EndIf
            \State\Return $\textsf{False}$
        \EndIf

        % \Assert{\textsf{ch\_peak.size} > \textsf{sum}(re\_peak\_sizes)}
    
        % \If{$|\textsf{re\_peak\_roots}| == 0$}
        % \Yield{\textsf{`Done'}}
        % \State\Return
    
        % \ElsIf{$\textsf{ch\_peak.size}\ //\ 2 > \textsf{re\_peak\_sizes}[0]$}
        %     \CommentLine{All responder peaks are still within challenger's left subtree.}
        %     \YieldFrom{\textsf{identify\_ch\_subtree\_root}(\textsf{ch\_peak.left, re\_peak\_roots, re\_peak\_sizes, peak\_number})}
        %     \State\Return
        
        % \ElsIf{$\textsf{ch\_peak.left.root} == \textsf{re\_peak\_roots}[0]$}
        %     \YieldFrom{\textsf{identify\_ch\_subtree\_root}(\textsf{ch\_peak.right, re\_peak\_roots[1:], re\_peak\_sizes[1:], peak\_number+1})}
        %     \State\Return
        
        % \Else
        %     \State\Return \YieldFrom{\textsf{challenge\_re\_tree}(\textsf{ch\_first\_peak, re\_first\_peak\_roots, re\_first\_peak\_size, [peak\_number]})}
        %     \State\Return
        % \EndIf
    \EndFunction
  \vskip2pt
  \end{algorithmic}
\end{algorithm*}

\begin{algorithm*}
  \caption{\label{alg.generate.commitments} Generation of the commitments by the \pos nodes. Returns the calculated commitment.}
  \begin{algorithmic}[3]
    \Function{\sf generate\_commitment}{\sf data, type}
        
        \If{$\textsf{type} == \textsf{checkpoint}$}
            \Let{\textsf{block\_headers}, \textsf{transaction\_roots}, \textsf{block\_bodies} }{\textsf{parse\_block\_data}(\textsf{data})}
            \State\Return $H(\textsf{block\_headers}\ ||\ \textsf{transaction\_roots})$
        \ElsIf{$\textsf{type} == \textsf{message}$}
            \State\Return $H(\textsf{data})$
        \EndIf
    \EndFunction
  \vskip2pt
  \end{algorithmic}
\end{algorithm*}

\begin{algorithm*}
  \caption{\label{alg.pos.validate.commitments} Validation of the commitments on \chain by the full \pos nodes. Returns true if the commitment is valid given the data $D$.}
  \begin{algorithmic}[3]
    
    \Function{\sf pos\_validate\_commitment}{\sf commitment, data, type}
    
        \If{$\textsf{type} == \textsf{checkpoint}$}
            \CommentLine{Parse the data for the checkpoint.}
            \Let{\textsf{block\_headers}, \textsf{transaction\_roots}, \textsf{block\_bodies} }{\textsf{parse\_block\_data}(\textsf{data})}
            \For{$i \gets 0,..,len(\textsf{block\_headers})-1$} 
                \Let{\sf block\_header}{ \textsf{block\_headers}[i]}
                \Let{\sf transaction\_root}{\textsf{transaction\_roots}[i]}
                \Let{\sf  block\_body}{\textsf{block\_bodies}[i]}
               \CommentLine{Check if the checkpointed PoS blocks were finalized.} \If{!$\textsf{is\_finalized}(\textsf{block\_header})$}
                    \State\Return $\textsf{False}$
                \EndIf
                \CommentLine{Check if the checkpointed PoS blocks form a consecutive sequence on the PoS chain.}
                \If{$i \neq len(\textsf{block\_headers})\  \land\ !\textsf{block\_headers}[i+1].\textsf{is\_ancestor}(\textsf{block\_header\ ||\ block\_body})$}
                    \State\Return $\textsf{False}$
                \EndIf
               \CommentLine{Check if the checkpoint was calculated with the correct transaction roots and if these roots commit to the PoS block body.} \If{$\textsf{transaction\_root} \neq  \textsf{block\_header}.\textsf{transaction\_root}\ \lor \ \textsf{transaction\_root} \neq H(\textsf{block\_body})$}
                    \State\Return $\textsf{False}$
                \EndIf
            \EndFor
            \CommentLine{Check if checkpoint was correctly calculated.}
            \If{$\textsf{commitment} == H(\textsf{block\_headers}\ ||\ \textsf{transaction\_roots})$}
                \State\Return $\textsf{True}$
            \EndIf
            \State\Return $\textsf{False}$
        \ElsIf{$\textsf{type} == \textsf{message}$}
            \CommentLine{Check if the message commitment was correctly calculated.}
            \If{$\textsf{commitment} == H(\textsf{data})$}
                \State\Return $\textsf{True}$
            \EndIf
            \State\Return $\textsf{False}$
        \EndIf

    \EndFunction
  \vskip2pt
  \end{algorithmic}
\end{algorithm*}

\begin{algorithm*}
  \caption{\label{alg.find.canonical.chain} Identifying the canonical PoS chain when there is a safety violation on the PoS chain. Returns the canonical PoS chain.}
  \begin{algorithmic}[3]
    \Function{\sf identify\_PoS\_chain}{\sf \chainnosp\_chain, PoS\_blocktree}
        \Let{\sf PoS\_canonical}{[]}
        \CommentLine{Obtaining a sequence of valid checkpoints and the associated data from the PoW chain. Note that if there are two consecutive valid checkpoints on \chain such that the second checkpoint commits to blocks conflicting with those of the first one, the second one is not returned.}
        \Let{\textsf{commitments}, \textsf{data}}{\textsf{\chainnosp\_chain}.\textsf{get\_valid\_checkpoints}()}
        \Let{\sf children}{[\textsf{PoS\_blocktree.genesis\_block}]}
        \Let{\sf next}{\sf True}
        \CommentLine{$i$ keeps track of which checkpoint in the sequence of returned valid checkpoints PoS node is currently considering.}
        \Let{i}{0}
        \If{$i \geq len(\textsf{commitments})$}
            \Let{\textsf{block\_headers}, \textsf{transaction\_roots}, \textsf{block\_bodies} }{\bot,\bot,\bot}   
        \Else
            \Let{\sf cur\_commitment}{\textsf{commitments}[0]}
            \Let{\sf cur\_data}{\textsf{data}[0]}
            \Let{\textsf{block\_headers}, \textsf{transaction\_roots}, \textsf{block\_bodies} }{\textsf{parse\_block\_data}(\textsf{cur\_data})}
        \EndIf
        \While{$\textsf{children} \neq \emptyset$}
            \For{$\textsf{PoS\_block} \gets \textsf{children}$}
               \CommentLine{Check if the current PoS block is committed by a valid and early checkpoint, and should be part of the canonical PoS chain.}
                \If{$\textsf{PoS\_block}.\textsf{header} \in \textsf{block\_headers}$}
                    \Let{\sf children}{\textsf{PoS\_blocktree}.\textsf{get\_children}(\textsf{PoS\_block})}
                    \Let{\sf next}{\sf False}
                    \Let{\sf PoS\_canonical}{\textsf{PoS\_canonical} + [\textsf{PoS\_block}]}
                    \CommentLine{All of the blocks in the current checkpoint are accounted for. PoS node now considers the PoS blocks attested by the next valid checkpoint.}
                    \If{\textsf{PoS\_block}.\textsf{header} == block\_headers[-1]}
                        \Let{i}{i+1}
                        \If{$i \geq len(\textsf{commitments})$}
                            \Let{\textsf{block\_headers}, \textsf{transaction\_roots}, \textsf{block\_bodies} }{\bot,\bot,\bot}   
                        \Else
                            \Let{\sf cur\_commitment}{\textsf{commitments}[i]}
                            \Let{\sf cur\_data}{\textsf{data}[i]}
                            \Let{\textsf{block\_headers}, \textsf{transaction\_roots}, \textsf{block\_bodies} }{\textsf{parse\_block\_data}(\textsf{cur\_data})}
                        \EndIf
                    \EndIf
                    \Break
                \EndIf
            \EndFor
            \CommentLine{If there is no checkpoint to help decide which children of a block on the PoS chain to follow as the next block on the canonical PoS chain, decision is made in favor of the first child returned by the $\textsf{get\_children}$ function.}
            \If{$\textsf{next}$}
                \Let{\sf PoS\_canonical}{\textsf{PoS\_canonical} + [\textsf{children}[0]]}
                \Let{\sf children}{\textsf{PoS\_blocktree}.\textsf{get\_children}(\textsf{children}[0])}
            \EndIf
            \Let{\sf next}{\sf True}
        \EndWhile
        \State\Return $\textsf{PoS\_canonical}$
    \EndFunction
  \vskip2pt
  \end{algorithmic}
\end{algorithm*}

\begin{algorithm*}
  \caption{\label{alg.stake.withdrawal.and.slashing} Granting stake withdrawal request and slashing protocol violators in the event of a safety violation. Returns true if the withdrawal request for the specified validator is to be granted in the view of the PoS node running the function, slashes the stake of the validator if there is a fraud proof accusing it.} 
  \begin{algorithmic}[3]
    \Function{\sf grant\_withdrawal\_request}{\sf \chainnosp\_chain, validator}
        \Let{\textsf{commitments}, \textsf{data}}{\textsf{\chainnosp\_chain}.\textsf{get\_valid\_checkpoints}()}
        \CommentLine{$\textsf{get\_all\_checkpointed\_blocks}$ returns all checkpointed, valid and finalized PoS blocks given a sequence of commitments and data.}
        \Let{\textsf{pos\_blocks} }{\textsf{get\_all\_checkpointed\_blocks}(\textsf{commitments, data})}
        \Let{\sf PoS\_block}{\bot}
        \For{$\textsf{block} \gets \textsf{blocks}$}
            \If{$\exists\ \textsf{withdrawal\_req} \text{ by } \textsf{validator} \in \textsf{block}$}
                \Let{\sf PoS\_block}{block}
                \Let{\sf \chainnosp\_block\_height}{\text{Height of the \bpow block containing the checkpoint for }\textsf{PoS\_block}}
            \EndIf
        \EndFor
        \CommentLine{Without a checkpoint for the PoS block containing the withdrawal request, request cannot be granted.}
        \If{$\textsf{PoS\_block} == \bot$}
            \State\Return $\mathsf{False}$
        \EndIf
        \CommentLine{If the checkpoint of the PoS block containing the request is not sufficiently deep in \chain in the view of the PoS node, then the request is not granted.}
        \If{$len(\textsf{\chainnosp\_chain}) \leq \textsf{\chainnosp\_block\_height} + k_w$}
            \State\Return $\mathsf{False}$
        \EndIf
        \CommentLine{If the checkpoint of the PoS block containing the request is indeed $k_w$ deep in \chain in the view of the PoS node, then the request is granted only after checking for the fraud proofs as specified by condition (3) in Section~\ref{sec:slashing-in-the-case-of-safety-violations}}
        \Let{\sf \chainnosp\_fragment}{\textsf{\chainnosp\_chain}[\textsf{\chainnosp\_block\_height}:\textsf{\chainnosp\_block\_height}+k_w]}
        \Let{\textsf{commitments}, \textsf{fraud\_proofs}}{\textsf{\chainnosp\_fragment}.\textsf{get\_valid\_fraud\_proofs}()}
        \For{$\textsf{fraud\_proof} \gets \textsf{fraud\_proofs}$ }
            \If{$\textsf{fraud\_proof} \text{ accuses } \textsf{validator}$}
                \CommentLine{Validator is slashed if there is a valid fraud proof on Babylon that irrefutably accuses the validator.}
                \State $\textsf{slash\_validator}(\textsf{validator}, \textsf{fraud\_proof})$
                \State\Return $\mathsf{False}$
            \EndIf
        \EndFor
        \State\Return $\mathsf{True}$
    \EndFunction
  \vskip2pt
  \end{algorithmic}
\end{algorithm*}

\begin{algorithm*}
  \caption{\label{alg.censorship.slashing} Identifying all censoring PoS blocks with respect to a censorship complaint on \chain. Takes as input the canonical \chain chain, height of the \chain block with the censorship complaint, and the canonical PoS chain. Slashes the validators that proposed or voted for the censoring blocks.} 
  \begin{algorithmic}[3]
    \Function{\sf is\_block\_censoring}{\sf \chainnosp\_chain \chainnosp\_block\_height, PoS\_block, PoS\_canonical}
        \Let{\textsf{commitments},\textsf{censored\_txs}}{\textsf{\chainnosp\_chain}[\textsf{\chainnosp\_block\_height}].\textsf{get\_valid\_censorship\_complaint()}}
        \CommentLine{Get all checkpoints that extend the \chain block with the censorship complaint by at least $2k_c$ blocks.}
        \Let{\textsf{commitments}, \textsf{data}}{\textsf{\chainnosp\_chain}[\textsf{\chainnosp\_block\_height}+2k_c:].\textsf{get\_valid\_checkpoints}()}
        \CommentLine{Parse over these checkpoints starting from the second one as stated in Section~\ref{sec:censorship-resilience}}
        \For{$i \gets 1,..,len(\textsf{data})-1$}
            \Let{\textsf{block\_headers}, \textsf{transaction\_roots}, \textsf{block\_bodies} }{\textsf{parse\_block\_data}(\textsf{data})}
            \CommentLine{Detect the censoring blocks on the PoS chain as attested by each checkpoint.}
            \For{$\textsf{PoS\_block} \gets \textsf{block\_headers},\textsf{block\_bodies}$}
                \If{$\textsf{censored\_txs} \notin \textsf{PoS\_canonical}[:\textsf{PoS\_block}]\ \land\ \textsf{censored\_txs} \text{ valid w.r.t the state of } \textsf{PoS\_canonical}[:\textsf{PoS\_block}]$}
                    \State $\text{Slash validators that proposed and voted upon } \textsf{PoS\_block}$
                \EndIf
            \EndFor
        \EndFor
    \EndFunction

  \vskip2pt
  \end{algorithmic}
\end{algorithm*}

\begin{algorithm*}
  \caption{\label{alg.stalling.slashing} Emulating a Tendermint round on \chain when there is a stalling evidence. This function is evoked by validators if a stalling evidence is observed on the $k_w/2$-deep prefix of the longest \chain chain. There are $n=3f+1$ active validators in total.} 
  \begin{algorithmic}[3]

    \Function{\sf emulate\_round\_on\_Babylon}{}
    
        \Upon{$\textsf{stalling\_evidence} \in \textsf{\chainnosp\_block}$}
            \Let{\sf checkpoint\_height}{\textsf{\chainnosp\_chain}.\textsf{get\_last\_checkpoints\_height}()}
            \Let{\sf commitment, data}{\textsf{\chainnosp\_chain}.\textsf{get\_last\_checkpoint}()}
            \Let{\textsf{block\_headers}, \textsf{transaction\_roots}, \textsf{block\_bodies} }{\textsf{parse\_block\_data}(\textsf{data})}
           
           \CommentLine{Stalling evidence is taken into consideration only if it comes to \chain at least $2k_c$ blocks after the last checkpoint on \chain.} \If{$\textsf{\chainnosp\_block.height} \geq \textsf{checkpoint\_height} + 2k_c$}
                \CommentLine{If the validator has observed new PoS blocks finalized at the heights not covered by the last checkpoint, it sends a checkpoint for the new PoS blocks instead of emulating the round on \chain.}
                \If{$\textsf{last\_finalized\_PoS\_block.header} \notin \textsf{block\_headers}$}
                    \State $\text{Send checkpoint to \chain}$
                \Else 
                    \CommentLine{See Appendix~\ref{sec:details-of-stalling-resilience} for more information on how proposal messages are selected and structured.}
                    \State $\text{Send Tendermint proposal to \chain}$
                \EndIf
            \EndIf
        \EndUpon
    
        \CommentLine{This is block $b_1$ on Figure~\ref{fig:stalling}.}
        \Upon{$\text{\chain block with height } \textsf{checkpoint\_height} + 3k_c$}
            \Let{\sf proposals}{\textsf{\chainnosp\_chain}[-k_c:].\textsf{get\_non\_censoring\_proposals}()}
            \If{$len(\textsf{proposals}) < 2f+1$}
                \State $\text{Slash validators with missing or censoring proposals}$
                \State $\text{Send new stalling evidence}$
                \State\Return
            \Else
                \CommentLine{Proposal with the largest $\textsf{validRound}$ is selected, details are in Appendix~\ref{sec:details-of-stalling-resilience}.}
                \Let{\sf round\_proposal}{\textsf{select\_round\_proposal}(\textsf{proposals})}
                \CommentLine{See Appendix~\ref{sec:details-of-stalling-resilience} for more information on how prevotes and precommits are structured.}
                \State $\text{Send prevote and precommit to \chain for } \textsf{round\_proposal}$
            \EndIf
        \EndUpon
    
        \CommentLine{This is block $b_2$ on Figure~\ref{fig:stalling}.}
        \Upon{$\text{\chain block with height } \textsf{checkpoint\_height} + 4k_c$}
            \Let{\sf prevotes}{\textsf{\chainnosp\_chain}[-k_c:].\textsf{get\_prevotes}(\textsf{proposal})}
            \Let{\sf precommits}{\textsf{\chainnosp\_chain}[-k_c:].\textsf{get\_precommits}(\textsf{proposal})}
            \If{$len(\textsf{prevotes}) < 2f+1\ \lor\ len(\textsf{precommits}) < 2f+1$}
                \State $\text{Slash validators whose prevotes and precommits are missing for } \textsf{round\_proposal}$
                \State $\text{Send new stalling evidence}$
                \State\Return
            \Else
                \State $\text{Finalize } \textsf{round\_proposal}$
            \EndIf
        \EndUpon
        
    \EndFunction
    
  \vskip2pt
  \end{algorithmic}
\end{algorithm*}

\end{document}